\documentclass[10pt,journal,compsoc]{IEEEtran}

\usepackage{balance,epsfig,subfigure,listings,epstopdf,setspace,fancyhdr,amsmath,amsthm,cite,url}
\usepackage{enumitem}
\usepackage{algorithm}
\usepackage{algpseudocode}

\makeatletter
\def\BState{\State\hskip-\ALG@thistlm}
\makeatother

\algrenewcommand\algorithmicindent{0.8em}%

\newtheorem{lemma}{Lemma}

\newtheorem{theorem}{Theorem}

\newcounter{linectr}

\allowdisplaybreaks[4]

\begin{document}
%
\title{Coflow Scheduling in Data Centers: Routing and Bandwidth Allocation}

\author{Li~Shi,
		Junwei~Zhang,
        Yang~Liu,
        and~Thomas~Robertazzi,~\IEEEmembership{Fellow,~IEEE}
\IEEEcompsocitemizethanks{
\IEEEcompsocthanksitem Li Shi is with Snap Inc. The work was performed when he was with Stony Brook University.
E-mail: lishi.pub@gmail.com.
\IEEEcompsocthanksitem Yang Liu is with Uber Inc.
E-mail: yangliu89415@gmail.com.
\IEEEcompsocthanksitem Junwei Zhang is with Uber Inc.
E-mail:junweizhang23@gmail.com.
\IEEEcompsocthanksitem Thomas Robertazzi is with the Department
of Electrical and Computer Engineering, Stony Brook University, Stony Brook, NY, 11794.
E-mail: thomas.robertazzi@stonybrook.edu.
}
\thanks{This work has been submitted to the IEEE for possible publication. Copyright may be transferred without notice, after which this version may no longer be accessible. }}



\IEEEtitleabstractindextext{%
\begin{abstract}
In distributed computing frameworks like MapReduce, Spark, and Dyrad, a coflow is a set of flows transferring data between two stages of a
job.
The job cannot start its next stage unless all flows in the coflow finish.
To improve the execution performance of such a job, it is crucial to reduce the completion time of a coflow which can contribute more than 50\% of
the job completion time.
While several schedulers have been proposed, we observe that routing, as a factor greatly impacting the Coflow Completion Time (CCT), has
not been well considered.

In this paper, we focus on the coflow scheduling problem and jointly consider routing and bandwidth allocation. 
We first provide an analytical solution to the problem of optimal bandwidth allocation with pre-determined routes.
We then formulate the coflow scheduling problem as a Mixed Integer Non-linear Programming problem and present its relaxed convex optimization problem.
We further propose two algorithms, CoRBA and its simplified version: CoRBA-fast, that jointly perform routing and bandwidth allocation for a given
coflow while minimizes the CCT.
Through both offline and online simulations, we demonstrate that CoRBA reduces the CCT by 40\%-500\% compared to the state-of-the-art algorithms. 
Simulation results also show that CoRBA-fast can be tens of times faster than all other algorithms with around 10\% performance degradation compared
to CoRBA, which makes the use of CoRBA-fast very applicable in practice. 
\end{abstract}

}

\maketitle

\IEEEdisplaynontitleabstractindextext

%
\IEEEpeerreviewmaketitle

\section{Introduction}
\label{sec:intro}
In recent years, we have witnessed a significant improvement on the IT infrastructures, like high performance computing systems, ultra high-speed
networks and large-scale storage systems. 
Benefiting from these improvements, our ability of collecting, storing and processing data has also been dramatically enhanced. 
In a data center owned by big corporations like Google or Twitter, every day, hundreds terabytes of data can be transferred into its data storage
system~\cite{bigdata:lin2013scaling,network:singh2015jupiter} 
and processed/analyzed by some computing frameworks such as MapReduce~\cite{mapreduce} and Spark~\cite{spark}.
By applying such data analysis on many different areas like physics, biology, medicine, manufacture and finance, we have greatly changed the world we
live in. 
In such a context, one of the most important goals pursued by engineers and researchers is improving the execution performance of those
data processing jobs.
 
To achieve this goal, a critical problem to solve is how to optimize data transferring time.
In many computing frameworks, jobs consist of a sequence of processing stages.
Between two consecutive stages, there is usually a set of flows which move output data of the previous stage to the nodes executing
the later stage.
A job cannot start its next stage until all flows in this set finish. 
We usually refer such a set of flows as a {\em coflow}~\cite{coflow:2012coflow}.
Since the transfer of a coflow can occupy more than 50\% of the job completion time~\cite{coflow:2011orchestra}, optimizing the Coflow Completion Time
(CCT) is important for improving the execution performance of jobs.

\begin{figure}[!t]
\begin{center}
\includegraphics[width=3in]{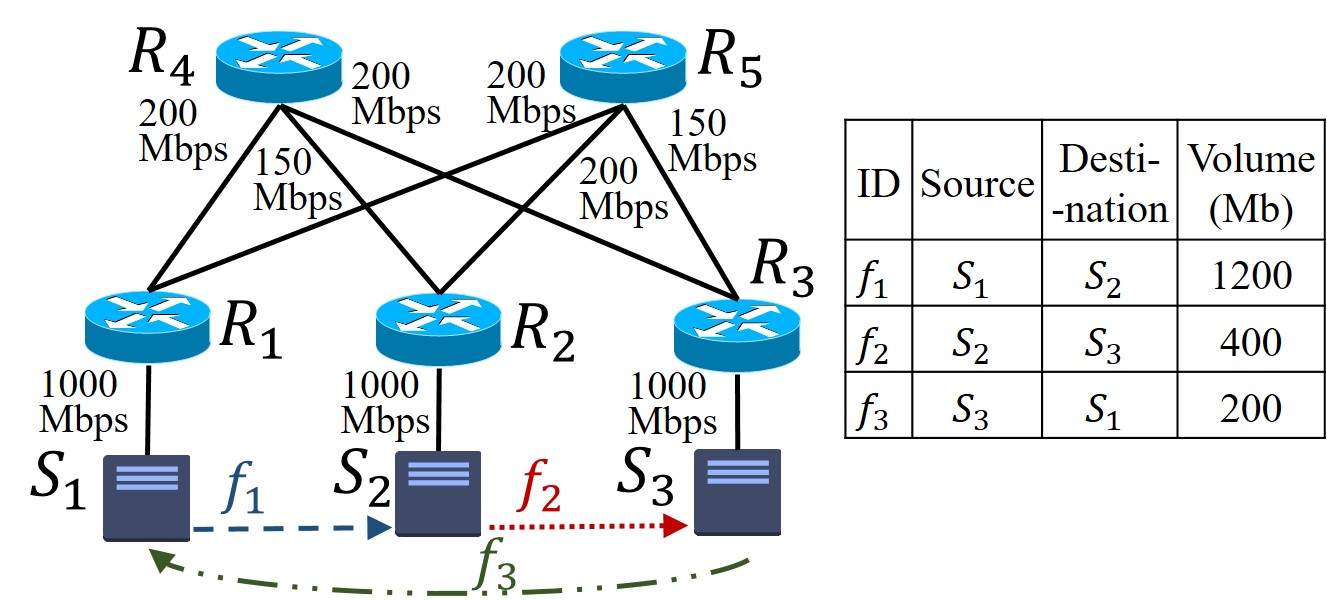}
\vspace{-2mm}
\caption{Scheduling 3 flows in a network with 3 servers and 5 routers.}
\label{fig:example}
\vspace{-2mm}
\end{center}
\end{figure}

Recently many mechanisms~\cite{popa2012faircloud,elasticswitch,chattycloud} have been proposed to provide
bandwidth guarantee to network flows. 
With such ability, we can guarantee the flow completion time (FCT), i.e., the completion time of a single flow. 
However, scheduling a coflow in a fashion that minimizes its completion time is
still a complex problem which involves both routing and bandwidth allocation at the level of the whole set of flows. 

To illustrate this problem, consider the scenario shown in Fig.~\ref{fig:example}, in which a coflow with 3 individual flows
$(f_1, f_2, f_3)$ is waiting to be scheduled in a network with 3 servers and 5 routers.
The goal is minimizing the CCT. 
To schedule this coflow, there is a need to determine a route for each flow and allocate a certain amount of bandwidth along each route. 
Fig.~\ref{fig:example-sol} shows three schedules generated by different scheduling strategies.
Fig.~\ref{fig:example1-sol1} shows a schedule in which each flow is routed via the maximum capacity path and bandwidth is fairly allocated to flows
using the same link, 
while Fig.~\ref{fig:example1-sol2} shows a schedule using the same routes but allocating bandwidth based on flow volume.
We can see that by appropriately allocating bandwidth, the completion time of the largest flow $f_1$ is successfully reduced, which leads a decreases
on the CCT.
However, the schedule shown in Fig.~\ref{fig:example1-sol2} is not the optimal one: The determined routes share the same link, which causes
bandwidth competition and limits the CCT.
This is because routing is performed at the level of individual flows rather than the whole set of flows in this schedule.
Whereas, Fig.~\ref{fig:example1-sol3} shows the optimal schedule in which we co-schedule all flows, reduce route overlap and allocate more bandwidth.

\begin{figure}[!t]
\begin{center}
\subfigure[Schedule 1: Maximum capacity path + Fair sharing]{
\includegraphics[width=3in]{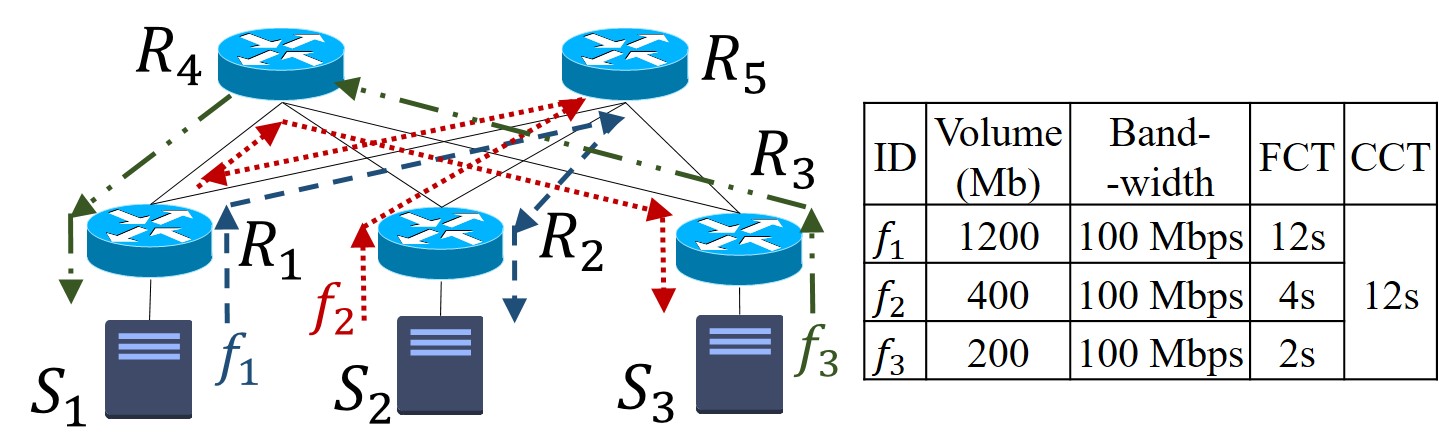}
\label{fig:example1-sol1}
}
\subfigure[Schedule 2: Maximum capacity path + Volume-proportional sharing]{
\includegraphics[width=3in]{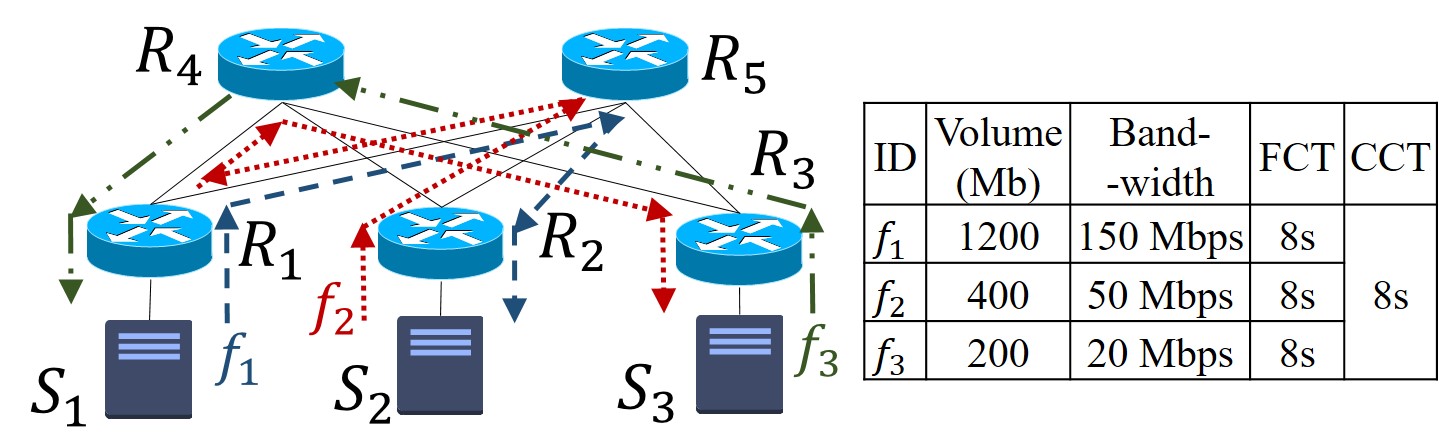}
\label{fig:example1-sol2}
}
\subfigure[Schedule 3: Jointly consider routing and bandwidth allocation]{
\includegraphics[width=3in]{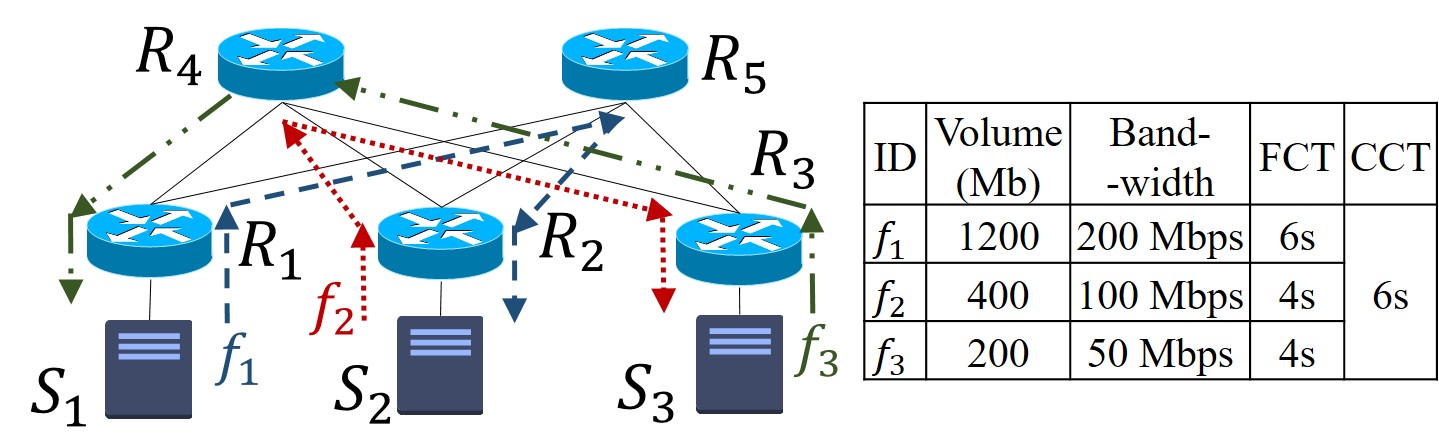}
\label{fig:example1-sol3}
}
\vspace{-2mm}
\caption{Three flow schedules generated by different strategies.}
\label{fig:example-sol}
\vspace{-4mm}
\end{center}
\end{figure}

From this example, we can see that routing and bandwidth allocation together determine the CCT. 
We can obtain the minimum CCT only if we find out the optimal solution on both routing and bandwidth allocation.
However, in practice, we may need to schedule a large number of flows in a network with thousands or tens of thousands of nodes. 
In such cases, there exists a vast amount of possible routing plans and for each routing plan there are very many ways to allocate bandwidth. 
Searching the optimal solution in such a huge solution space is not an easy problem to solve.

While several approaches have been proposed to schedule
coflows~\cite{coflow:2011orchestra,coflow:2012coflow,coflow:2014barrat,coflow:2014varys,coflow:2015aalo,coflow:2015rapier},
they either do not jointly consider routing and bandwidth
allocation~\cite{coflow:2011orchestra,coflow:2012coflow,coflow:2014barrat,coflow:2014varys,coflow:2015aalo} or perform routing based on a limited set
of candidate paths~\cite{coflow:2015rapier}.
In this paper, we focus on the coflow scheduling problem in which we jointly consider routing and bandwidth allocation for a given coflow.
In summary, our main contributions include
\begin{itemize}[align=left,leftmargin=*]
  \item We study the problem of optimal bandwidth allocation with pre-determined routes, 
  formulate it as a convex optimization problem, and provide an analytical solution. 
  By solving this problem, we essentially reduce the dimension of the coflow scheduling problem:
  For any routing plan, we can always optimally allocate bandwidth. (Section~\ref{sec:opt-ba}) 
  \item We formulate the Coflow Scheduling (CoS) problem as a Mixed Integer Non-linear Programming (NLMIP) problem incorporating both
  routing constraints and bandwidth allocation constraints.
  We then present a relaxation and an equivalent solvable convex optimization problem. (Section~\ref{sec:cos})
  \item We propose an algorithm, called Coflow Routing and Bandwidth Allocation (CoRBA), that solves the CoS problem. 
  With more practical consideration, we further propose CoRBA-fast, a simplified version of CoRBA. (Section~\ref{sec:corba})
  \item We evaluate the performance of CoRBA and CoRBA-fast by comparing them with the state-of-the-art algorithms through both offline and online
  simulations.
  The simulation results show that CoRBA can achieve 40\%-500\% smaller CCT than the existing algorithms. 
  The results also show that CoRBA-fast can be hundreds times faster than CoRBA with up to 10\% performance degradation, which makes it a good choice
  in practical use. (Section~\ref{sec:sim})
  \item We further discuss about the applicability of CoRBA-fast by comparing its running time with the transferring time of coflows in some typical
  MapReduce jobs.
  (Section~\ref{sec:discuss})
  
\end{itemize}


\section{Related Work}
\label{sec:related}
A significant amount of research has focused on the area of flow transmissions in data center networks. 
In this section, we discuss some of the works most relevant to our problems.

\smallskip \noindent {\bf Coflow schedulers.}
The problem of coflow-aware flow scheduling has attracted significant
attention~\cite{coflow:2012coflow,coflow:2011orchestra,coflow:2014barrat,coflow:2014varys,coflow:2015aalo,coflow:2015rapier}.
Orchestra~\cite{coflow:2011orchestra} is a global control architecture that manages the transfer of a set of correlated flows. 
Chowdhury {\em et al.}~\cite{coflow:2012coflow} propose the concept of coflow, a network abstraction that describes the traffic pattern of
prevalent data flows.
Barrat~\cite{coflow:2014barrat} is a decentralized flow scheduler which groups flows of a task and schedules them together with scheduling policies
like FIFO-LM.
Varys~\cite{coflow:2014varys} is a coflow scheduler that addresses the inter-coflow scheduling problem, supports deadlines and guarantees coflow
completion time. 
Aalo~\cite{coflow:2015aalo} is a recently proposed coflow scheduler that schedules coflows without any prior knowledge and supports pipelining and
dependencies in multi-stage DAGs.
While the coflow schedulers introduced above focus on scheduling individual coflows or groups of coflows, however, they neglect routing, an important
factor that impacts the coflow completion time. 

RAPIER~\cite{coflow:2015rapier} is a recently proposed coflow-aware network scheduling framework that integrates both routing and bandwidth
allocation. 
In RAPIER, scheduling a single coflow is formulated as a linear programming problem in which the route of each flow in the coflow is selected from a
set of candidate paths given as input. 
In contrast to RAPIER, in this paper, we propose algorithms that consider the bandwidth availability of the whole network and route flows via the best
paths instead of picking up a path from a set of candidates. 
We show the superior performance of our algorithms by comparing them with the optimization algorithm used in RAPIER.

\smallskip \noindent {\bf Flow schedulers.}
Much research work has also been performed on reducing the average flow completion
time~\cite{flowte:rojas2015schemes,flow:2011d3,flow:2012pdq,flow:2013pfabric,flow:2014repflow,flow:2015pias}.
Rojas {\em et al.}~\cite{flowte:rojas2015schemes} give a comprehensive survey on existing schemes for scheduling flows in data center
networks.
PDQ~\cite{flow:2012pdq} is a flow scheduling protocol which utilizes explicit rate control to allocate bandwidth to flows and enables flow preemption. 
pFabric~\cite{flow:2013pfabric} is a datacenter transport design that decouples flow scheduling from rate control, 
in which flows are prioritized and switches implement a very simple priority-based scheduling/dropping mechanism. 
RepFlow~\cite{flow:2014repflow} is a transport design that replicates each short flow. 
It transmits the replicated and original flows via different paths, which reduces the probability of experiencing long
queueing delay and therefore decreases the flow completion time.
PIAS~\cite{flow:2015pias} is an information-agnostic flow scheduling scheme minimizing the FCT by mimicking the Shortest Job First strategy
without any prior knowledge of incoming flows. 
While these existing schemes can reduce the FCT by using different strategies, they are not coflow-aware and therefore have different optimization
objective compared to our algorithms.


\section{Problem Definition}
\label{sec:definition}
In this paper, we consider scheduling a given coflow on a given data center network. We now introduce the input, output, and objective of this
problem.

\smallskip \noindent {\bf Input.} 
The input contains a data center network and a coflow. 
The network is modeled as a graph $\mathcal G = <\mathcal V,\mathcal E,\mathcal
B>$, where $\mathcal V$ is the set of nodes, in which each server and router corresponds to one node; 
$\mathcal E$ is the set of links, in which $E_{uv}$ presents the link between node $u$ and node $v$;
and $\mathcal B$ is the set of available bandwidth on links, in which $B_{uv}$ presents the available bandwidth of $E_{uv}$.
Note that each link in $\mathcal E$ is unique, i.e., the link between $u$ and $v$ is modeled as either
$E_{uv}$ or $E_{vu}$.

The coflow is a set of flows. We denote it by $\mathcal{CF} = \{ F_1, F_2, \ldots, F_N\}$ and define each flow $F_i$ as $\{ S_i, D_i, V_i \}$,
where $S_i$/$D_i$ is the source/destination node of this flow and $V_i$ is the data volume, 
i.e., the total amount of data to be transferred.
Like prior works~\cite{coflow:2014barrat,coflow:2015aalo,coflow:2014varys,coflow:2015rapier},
we assume that the information of a coflow can be captured by upper layer applications~\cite{coflow:2012coflow} or using existing prediction
techniques~\cite{network:2014hadoopwatch}.

\smallskip \noindent {\bf Output.}
The output contains a set of routes, one for each flow in $\mathcal{CF}$, and a certain amount of bandwidth allocated to each route.
We define the set of routes as $\mathcal{P}= \{ p_1, p_2,\ldots,p_N \}$ in which $p_i$ is the route selected for flow $F_i$.
We also define $b_i$ as the amount of bandwidth allocated to the route $p_i$.

\smallskip \noindent {\bf Objective.}
Our objective is minimizing the CCT that is the completion time of the last finished flow.
Let $ct_i$ denote the completion time of flow $F_i$, then our objective is
\begin{equation}    
\label{eqn:obj}
Minimize \ \ \max_{i=1,\ldots,N}\left\{ ct_i \ \big|\ ct_i = \frac{V_i}{b_i} \right\}. 
\end{equation}

\section{Optimal Bandwidth Allocation with Pre-determined Routes}
\label{sec:opt-ba}
We start from the problem of Optimal Bandwidth Allocation (OptBA) with pre-determined routes, 
in which the route of each flow has already been determined and we need to allocate bandwidth to these routes while minimizing the CCT. 
With the solution of this problem, we can optimally allocate bandwidth corresponding to any routing plan, which essentially reduces the dimension of
the coflow scheduling problem. 

To model this problem, we define $X^i_{uv}$ as a binary constant which has the following value 
\begin{gather}
\label{eqn:optba-x}
X^i_{uv} =
\left\{
	\begin{array}{ll}
		1,  & \mbox{if link } E_{uv} \mbox{ is on the path } p_i, \\		
		0,  & \mbox{otherwise}.
	\end{array}
\right.
\end{gather} 
We formulate the OptBA problem as
\smallskip \\  
{\bf OptBA \\}
%
%
{\bf Objective:} \\
\vspace{-3mm}
\begin{gather}
\label{eqn:optba-object}
Minimize \ \ \max_{i=1,\ldots,N}\left\{ ct_i \ \big|\ ct_i = \frac{V_i}{b_i} \right\} 
\end{gather}
\vspace{-5mm}
\\
{\bf Subject to}
\vspace{-2mm}
\begin{gather}
\label{eqn:optba-constraint1}
\sum_{i=1}^{N} X^i_{uv} \cdot b_i \leq B_{uv},\ \forall u, \forall v, that\ E_{uv} \in \mathcal E, 
\end{gather}
\vspace{-5mm}
\begin{gather}
\label{eqn:optba-constraint2}
b_i \geq 0, \  i= 1, \ldots ,N
\end{gather}
\vspace{-5mm}
\\
{\bf Remarks:}
\begin{itemize}
    \item Constraints (\ref{eqn:optba-constraint1}) are bandwidth availability constraints, which limit that for any link $E_{uv}$, the
    overall amount of bandwidth allocated to flows using this link cannot exceed the available bandwidth of this link.
    \item Constraints (\ref{eqn:optba-constraint2}) are domain constraints ensuring the bandwidth allocated to each
    flow to be non-negative.
\end{itemize}
{\bf Convexity.} 
We observe that the functions in constraints (\ref{eqn:optba-constraint1}) and (\ref{eqn:optba-constraint2}) are
all affine on {\bf $b_i$} and thus convex. 
At the same time, the function $ct_i$ is convex, because its second derivative is nondecreasing when $b_i$ is larger than 0.
Therefore, according to~\cite{convexOpt}, the objective function, i.e., the pointwise maximum function of $ct_i$, is
also a convex function.
As a result, the OptBA problem is a convex optimization problem.

\subsection{Analytical Solution}
While existing convex optimization algorithms can be used to solve the OptBA problem,
we develop an analytical solution which is more efficient.
Specifically, we define $b^i_{uv}$ as the amount of bandwidth allocated to flow $F_i$ on link $E_{uv}$ when we
proportionally distribute the available bandwidth of link $E_{uv}$ to all flows that are using this link, i.e.,
\begin{equation}
b^i_{uv} = \frac{V_i}{\sum^N_{k=1}X^k_{uv}V_k}B_{uv}.
\end{equation} 
Let $\vec{b^*}$ be an vector$\{b^*_1, b^*_2, \ldots , b^*_N \}$, in which  $b^*_i$ is the minimum value of all existing
$b^i_{uv}$, i.e.,
\begin{gather}
\label{eqn:dsba-opt-sol}
\begin{aligned}
\vec{b^*} = \{b^*_1, b^*_2, \ldots , b^*_N \},  \text{where } b^*_i = \min_{E_{uv} \in \mathcal E_i}\big\{ b^i_{uv} \big\},
\end{aligned}
\end{gather}
in which $\mathcal E_i$ is the set of links on the route of flow $F_i$.
We then have that the vector $\vec{b^*}$ is an optimal solution of the OptBA problem.
To prove this, we first prove the following lemma.
\begin{lemma}
\label{lemma:l0}
Assume that for the vector $\vec{b}^*$, flow $F_p$ has the largest finish
time, i.e., $ct^*_p=V_p/b^*_p =\max_{i=1,\ldots,N}\{ct^*_i\}$.
Also assume that 
\begin{gather}
\label{eqn:lemma1-1}
b^*_p = \min_{E_{uv} \in \mathcal E_p}\big\{ b^p_{uv} \big\} = b^p_{u^*v^*} =
\frac{V_p}{\sum^N_{k=1}X^k_{u^*v^*}V_k}B_{u^*v^*}.
\end{gather}
Then for every other flow $F_i$ using link $E_{u^*v^*}$, we have 
\begin{gather}
\label{eqn:lemma1-2}
b^*_i = \min_{E_{uv} \in \mathcal E_i}\big\{ b^i_{uv} \big\} = b^i_{u^*v^*} =
\frac{V_i}{\sum^N_{k=1}X^k_{u^*v^*}V_k}B_{u^*v^*}.
\end{gather}
\end{lemma}
\begin{proof}
To begin with, we assume that there exists a flow $F_q$ using link $E_{u^*v^*}$ but getting its allocated bandwidth when
another link $E_{u'v'}$ is considered, i.e.,
\begin{gather}
\label{eqn:lemma1-3}
b^*_q = \min_{E_{uv} \in \mathcal E_q}\big\{ b^q_{uv} \big\} = b^q_{u'v'} =
\frac{V_q}{\sum^N_{k=1}X^k_{u'v'}V_k}B_{u'v'}.
\end{gather}
Next, define $ct'_q$ as $V_q / b^q_{u^*v^*}$ and we have 
\begin{gather}
\label{eqn:lemma1-5}
ct^*_q = \frac{V_q}{b^q_{u'v'}} > \frac{V_q}{b^q_{u^*v^*}} = ct'_q.
\end{gather} 
On the other hand, based on Equation (\ref{eqn:lemma1-1}), we have
\begin{gather}
\label{eqn:lemma1-6}
ct^*_p = \frac{V_p}{b^p_{u^*v^*}} = \frac{V_q}{b^q_{u^*v^*}} = ct'_q.
\end{gather} 
Now using the Inequity (\ref{eqn:lemma1-5}) and Equation (\ref{eqn:lemma1-6}), we get
\begin{gather}
ct^*_q > ct'_q = ct^*_p,
\end{gather}
which conflicts with the assumption the flow $F_p$ has the largest completion time. 
Therefore, there does not exist a flow $F_q$ satisfying Equation (\ref{eqn:lemma1-3}).
We have proved the lemma now.
\end{proof}
\vspace{-1mm}
Based on Lemma~\ref{lemma:l0}, we have the following theorem.
\vspace{-1mm}
\begin{theorem}
The vector $\vec{b}^*$ defined in Equation~\ref{eqn:dsba-opt-sol} is an optimal solution of the OptBA problem.
\end{theorem}
\begin{proof}
Assume that flow $F_p$ has the largest completion time and $b^*_p$ obtains its value when link $E_{u^*v^*}$ is
considered. 
Then with Lemma~\ref{lemma:l0}, for every flow $F_i$ using $E_{u^*v^*}$, there exists
\begin{equation}
\label{eqn:thm1-1}
b^*_i = b^i_{u^*v^*} = \frac{V_i}{\sum^N_{k=1}X^k_{u^*v^*}V_k}B_{u^*v^*}. 
\end{equation}
Therefore, for each flow $F_i$, we have
\begin{gather}
\label{eqn:thm1-2}
ct^*_i = \frac{V_i}{b^*_i} = \frac{V_p}{b^*_p} = ct^*_p.
\end{gather}
Now assume that another vector $\vec{b'}$ is the optimal solution and flow $F_q$ has the largest completion time. 
We then have
\begin{gather} 
\label{eqn:thm1-4}
ct'_i \leq ct'_q < ct^*_p = ct^*_i,\ \forall i,\ that\ X^i_{u^*v^*} \neq 0.
\end{gather}
Naturally, we have $b'_i > b^*_i,\ \forall i,\ that\ X^i_{u^*v^*} \neq 0.$
On the other hand, from Equation (\ref{eqn:thm1-1}), we can get
\begin{gather}
\label{eqn:thm1-6}
\sum^{N}_{i=1} X^i_{u^*v^*}b^*_i = B_{u^*v^*}.
\end{gather}
Putting them together, we have
\begin{gather}
\sum^{N}_{i=1} X^i_{u^*v^*}b'_i > \sum^{N}_{i=1} X^i_{u^*v^*}b^*_i = B_{u^*v^*},
\end{gather}
which is infeasible. 
As a result, there does not exist a feasible solution better than $\vec{b^*}$. 
Therefore, the vector $\vec{b^*}$ defined in Equation (\ref{eqn:dsba-opt-sol}) is an optimal solution of the OptBA problem.
\end{proof}

\section{The Coflow Scheduling (CoS) Problem}
\label{sec:cos}
In this section, we formulate the CoS problem as a NLMIP problem, present a non-linear relaxation, and further transform the relaxation
to a solvable convex optimization problem.

We start with defining variable $x^i_{uv}$ as
an integer variable with three possible values (-1, 0, and 1):
\begin{gather}
\label{eqn:cos-x}
x^i_{uv} =
\left\{
	\begin{array}{ll}
		1,  & \mbox{if } F_i \mbox{ flows from node u to v}, \\
		-1, & \mbox{if } F_i \mbox{ flows from node v to u}, \\
		0,  & \mbox{otherwise}.
	\end{array}
\right.
\end{gather} 
Note that there exists either $x^i_{uv}$ or $x^i_{vu}$, corresponding to the existence of either $E_{uv}$ or
$E_{vu}$. We then have \\
{\bf CoS \\}
{\bf Variables}
\begin{itemize}
  \item $b_{i}$: bandwidth allocated to the path $p_i$ selected for $F_i$.
  \item $x^i_{uv}$: an integer variable defined by Equation (\ref{eqn:cos-x}).
  \item $ct_i$: the completion time of flow $F_i$.
\end{itemize}
{\bf Constants}
\begin{itemize}
  \item $N$: the number of flows in the set $\mathcal F$.  
  \item $B_{uv}$: the available bandwidth of the link $E_{uv}$. 
  \item $\mathcal N(u)$: the set of neighbor nodes of node $u$.  
\end{itemize}
{\bf Objective:} \\
\begin{gather}
\label{eqn:cos-object}
Minimize \ \ \max_{i=1,\ldots,N}\left\{ ct_i \ \big|\ ct_i = \frac{V_i}{b_i} \right\}
\end{gather}
\\
{\bf Subject to}
\begin{gather}
\label{eqn:cos-src}
\sum_{w \in \mathcal N(s_i)} x^i_{s_i w} - \sum_{w \in \mathcal N(s_i)} x^i_{w s_i} = 1, \ \ i=1,\ldots,N, \\
\label{eqn:cos-dst}
\sum_{w \in \mathcal N(s_i)} x^i_{w d_i} - \sum_{w \in \mathcal N(s_i)} x^i_{d_i w} = 1, \ \ i=1,\ldots,N, \\
\label{eqn:cos-flow}
\sum_{w \in \mathcal N(u)} x^i_{wu} - \sum_{w \in \mathcal N(u)} x^i_{uw} = 0, \ \forall i, \forall u \notin \{ s_i, d_i\}, \\ 
\label{eqn:cos-bw}
\sum^N_{i=1} \big| x^i_{uv}\cdot b_i \big| \leq B_{uv}, \ \forall u, \forall v, that\ E_{uv} \in \mathcal E, \\
\label{eqn:cos-domain-b}
b_i \geq 0, \  i= 1, \ldots ,N, \\
\label{eqn:cos-domain-x}
x^i_{uv} \in \{-1,0,1\}, \ \forall i,\ \forall u,\ \forall v.
\end{gather}
\\
{\bf Remarks:}
\begin{itemize}
    \item Constraints (\ref{eqn:cos-src}) ensure that data is sent out from the source of any flow through only one
    link, because a positive value of $x^i_{s_i w}$ and a negative value of $x^i_{w s_i}$ mean that data is going out
    from the source node via link $E_{s_i w}$ or $E_{w s_i}$. 
    Constraints (\ref{eqn:cos-dst}) ensure that data is transferred into the any destination node through only one link.
    \item Constraints (\ref{eqn:cos-flow}) ensure flow conservation, i.e., for any flow $F_i$ and any
    intermediate node $u$, the number of links via which data is transferred into node $u$ should be equal to the
    number of links via which data is sent out. Note that constraints (\ref{eqn:cos-src})-(\ref{eqn:cos-flow}) together enforce
    that only one path is selected to transfer data for each flow.
    \item Constraints (\ref{eqn:cos-bw}) enforce that the overall bandwidth allocated to flows in coflow $\mathcal {CF}$ on
    any link $E_{uv}$ does not exceed the total available bandwidth of that link.
    \item Constraints (\ref{eqn:cos-domain-b}) and (\ref{eqn:cos-domain-x}) are domain constraints. 
\end{itemize}

\subsection{A Relaxation of the CoS Problem and An Equivalent Convex Optimization Problem}
To solve the CoS problem, we first consider its relaxation in which the integer variable $x^i_{uv}$ is relaxed to a real variable, i.e., changing the
constraint (\ref{eqn:cos-domain-x}) to
\begin{equation}
-1 \leq x^i_{uv} \leq 1, \ \forall i,\ \forall u,\ \forall v.
\end{equation}
We name the relaxed problem as {\bf CoS-Relax}.

Subsequently, we transform the CoS-Relax problem into an equivalent convex optimization problem. 
We start from defining variable $T$ as the CCT, i.e., 
\begin{equation}
T = \max_{i=1,\ldots,N}\left\{ \frac{V_i}{b_i} \right\}. \nonumber
\end{equation}   
The objective function (\ref{eqn:cos-object}) is tranformed to
\begin{gather}
Minimize \ \ T \nonumber \\
\label{eqn:cos-liear-load-pre}
\frac{V_i}{b_i} \leq T , \ i=1,\ldots,N.
\end{gather}
We further define variable $q_i$ as
\begin{gather}
\label{eqn:cos-b2q}
q_i	 = T \cdot b_i, 
\end{gather}
By substituting $q_i$ into constraint (\ref{eqn:cos-bw}), we have
\begin{equation}
\label{eqn:cos-relax-eqv-bw}
\sum^N_{i=1} \big| x^i_{uv}\cdot q_i \big| \leq T \cdot B_{uv}, \ \forall u, \forall v, that\ E_{uv} \in \mathcal E. 
\end{equation}
Next, we define variable $p^i_{uv}$ as 
\begin{equation}
\label{eqn:cos-x2p}
p^i_{uv} = x^i_{uv} \cdot q_i.  
\end{equation}
By substituting $q_i$ and $p^i_{uv}$ into constraint (\ref{eqn:cos-liear-load-pre}), (\ref{eqn:cos-relax-eqv-bw}) and all other constraints, 
we get an equivalent problem of CoS-Relax, named CoS-Relax-Cvx, as shown below.
\\ {\bf CoS-Relax-Cvx \\}
\vspace{-2mm}
\begin{gather}
Minimize \ \ T
\end{gather}
\vspace{-5mm}
\\
{\bf Subject to}
\begin{gather}
q_i \geq V_i, \ i=1,\ldots,N, \\
\sum_{w \in \mathcal N(s_i)} p^i_{w s_i} - \sum_{w \in \mathcal N(s_i)} p^i_{s_i w} = -q_i \ \ i=1,\ldots,N, \\
\sum_{w \in \mathcal N(s_i)} p^i_{w d_i} - \sum_{w \in \mathcal N(s_i)} p^i_{d_i w} = q_i \ \ i=1,\ldots,N, \\
\sum_{w \in \mathcal N(u)} p^i_{wu} - \sum_{w \in \mathcal N(u)} p^i_{uw} = 0, \ \forall i, \forall u \notin \{ s_i, d_i\}, \\ 
\label{eqn:cos-linear-bw-pre}
\sum^N_{i=1} \big| p^i_{uv} \big| \leq T \cdot B_{uv}, \ \forall u, \forall v, that\ E_{uv} \in \mathcal E, \\
T \geq 0.
\end{gather}
All constraints in the CoS-Relax-Cvx problem are affine, except the constraints (\ref{eqn:cos-linear-bw-pre}) which are convex constraints,
because the absolute value of $p^i_{uv}$ is a convex function and the sum of convex functions is still convex.
Therefore, the CoS-Relax-Cvx problem is a convex optimization problem which can be solved by using convex optimization algorithms~\cite{convexOpt}.

\section{The CoRBA Algorithm and Its Simplified Version: CoRBA-fast}
\label{sec:corba}
To solve the CoS problem, we propose an algorithm, called Coflow Routing and Bandwidth Allocation (CoRBA).

The CoRBA algorithm has three phases:
(I) Obtaining a solution of the CoS-Relax problem by solving the CoS-Relax-Cvx problem; 
(II) Determining an initial solution of the CoS problem based on the
solution of the CoS-Relax problem;
(III) Utilizing a local search procedure to further optimize the obtained initial solution.
Algorithm~\ref{algo:corba} shows the pseudocode of CoRBA. 
We now introduce the details of each phase. 
 \\ \noindent
{\bf Phase I: Solve the relaxed problem.}
To begin with, CoRBA solves the CoS-Relax-Cvx problem by using convex algorithms~\cite{convexOpt}.
Let the solution be $T'$, $q'_i$, and ${p'}^{i}_{uv}$.

Subsequently, CoRBA calculates the solution of CoS-Relax, denoted by ${x'}^i_{uv}$ and $b'_i$, using Equations
(\ref{eqn:cos-b2q}) and (\ref{eqn:cos-x2p}). 

\begin{algorithm}[t]
\caption{The CoRBA Algorithm}
\label{algo:corba}
\begin{algorithmic}[1]
\Function{CoRBA}{\mbox{Network} $\mathcal G = <\mathcal V,\mathcal E,\mathcal B>$, \mbox{Coflow} $\mathcal{CF}$}
\Statex {\em Phase I:}
	\State Solve CoS-Relax-Cvx and get $\{T',\ q'_i,\ {p'}^{i}_{uv}\}$;
	\State Calculate ${x'}^i_{uv}$ and $b'_i$ using Equations (\ref{eqn:cos-b2q}) and (\ref{eqn:cos-x2p});
\Statex {\em Phase II:}	
	\For{each $F_i \in \mathcal{CF}$}		
		\State $p^{MC}_i \gets$ max capacity path using ${x'}^i_{uv}$ as link capacity;
		\State $x^i_{uv} \gets 1$ if $E_{uv}$ on $p^{MC}_i$; otherwise,  $x^i_{uv} \gets 0$; 
	\EndFor	
	\State Calculate $b_i$ using Equation (\ref{eqn:corba-bi});
	\State CCT $\gets max_{i=1,\ldots,N}\{ct_i\ |\ ct_i = V_i/b_i\}$;	
\Statex {\em Phase III:}
	\State Update $\mathcal B$ according to $x^i_{uv}$ and $b_i$;
	\While{true}
		\State $\mathcal F_{max} \gets \{ F_i \ |\ ct_i == CCT \}$;
		\For{each $F_i \in \mathcal F_{max}$}
			\State $\mathcal E_{congest} \gets \{ E_{uv}\ |\ E_{uv} \text{ is on } p_i\ \& \ B_{uv} == 0\}$;		
			\State Add $b_i$ back to $\mathcal B$ along route of $F_i$;
			\State Set each $E_{uv} \in \mathcal E_{congest}$ as unavailable;
			\State $p^{MC}_i \gets$ new max capacity path;
			\State Reset $x^i_{uv}$ according to new $p^{MC}_i$;
			\State Re-calculate $b_i$, $ct_i$, and CCT;
			\State {\bf if} $ct_i$ is reduced {\bf then break};
			\State {\bf else} Reverse all changes made for $F_i$;
		\EndFor
		\State {\bf if} no $ct_i$ is reduced {\bf then break};
	\EndWhile 
	\State {\bf return} $x^i_{uv}$, $b_i$, and CCT;
\EndFunction
\end{algorithmic}
\end{algorithm}
\noindent {\bf Phase II: Obtain an initial solution of the CoS problem.}
In this phase, the CoRBA algorithm obtains an initial solution of the CoS problem in two steps:
\begin{itemize}[align=left,leftmargin=*]
\item First, it determines the value of $x^i_{uv}$, i.e., obtaining a route for each flow.
Specifically, for each flow $F_i$, CoRBA uses ${x'}^i_{uv}$ as the capacity of $E_{uv}$ and find the max
capacity path~\cite{MaxCapaPath:punnen1991linear} (denoted by $p^{MC}_i$) between the source and destination of this flow. 
CoRBA then sets $x^i_{uv}$ as 1 for each link on that path and set its value as 0 for all other links, i.e.,
\begin{gather}
\label{eqn:corba-x}
x^i_{uv} =
\left\{
	\begin{array}{ll}
		1,  & \mbox{if } E_{uv} \mbox{ is on the path } p^{MC}_i, \\		
		0,  & \mbox{otherwise}.
	\end{array}
\right.
\end{gather}
\item Second, CoRBA determines the value of $b_i$, i.e., the amount of bandwidth allocated to each flow. 
Because the route of each flow is already determined, 
the CoS problem is naturally reduced to the OptBA problem which is already solved in section~\ref{sec:opt-ba}.
Therefore, according to Equation (\ref{eqn:dsba-opt-sol}), CoRBA calculates the value of $b_i$ as
\begin{gather}
\begin{aligned}
\label{eqn:corba-bi}
b_i = \min_{E_{uv} \in \mathcal E_i}\big\{ \frac{V_i}{\sum^N_{k=1}|x^k_{uv}|V_k}B_{uv}. \big\}.
\end{aligned}
\end{gather}
\end{itemize} \
\noindent {\bf Phase III: Local search.} 
In this phase, CoRBA utilizes a local search procedure to further improve the initial solution.

In this local search procedure, CoRBA runs in iterations. 
In each iteration, it starts with identifying all flows with the largest completion time and putting them into a
set named $\mathcal F_{max}$. 
Subsequently, CoRBA iteratively considers each flow in the set $\mathcal F_{max}$.
For each flow $F_i$ in $\mathcal F_{max}$, the CoRBA algorithm puts all links that are on the route of this flow and
that do not have any available bandwidth (Due to the initial bandwidth allocation to flows) into set $\mathcal E_{congest}$;
it then sets all links in $\mathcal E_{congest}$ as unavailable and find a new max capacity route for flow $i$;
in the following, it temporarily changes the route of flow $F_i$ to the new route and re-calculates $b_i$ and CCT for current routing plan. 
If the completion time of flow $F_i$ is reduced in the new solution, the CoRBA algorithm stops considering all
other flows in the set $\mathcal F_{max}$ and starts a new iteration;
otherwise, it reverts all temporary changes that it has made and moves to the next flow in the set $\mathcal F_{max}$. 

If CoRBA cannot improve the completion time of any flow in $\mathcal F_{max}$ in some iteration, it stops and outputs current
solution as the final solution, as it cannot improve the CCT anymore.

\subsection{CoRBA-fast: A Simplified Version of CoRBA}
The CoRBA algorithm begins with solving the CoS-Relax-Cvx problem which is a convex optimization problem. 
When the problem's scale is large enough, solving this problem can be time-consuming. 
On the other hand, we observe that the local search procedure in the CoRBA algorithm can be used to optimize any feasible coflow schedules.
With such observations, we propose CoRBA-fast, a simplified version of CoRBA with less time complexity.
 
CoRBA-fast has two phases: First, it obtains an initial solution; Second, it utilizes the local search
procedure used in the CoRBA algorithm to optimize the initial solution.
 
To obtain an initial solution, for each flow $F_i$, the CoRBA-fast algorithm set the route of this flow as the shortest maximum capacity path (i.e.,
the shortest path with the maximum capacity). 
Subsequently, CoRBA-fast calculates $b_i$ and CCT based on the determined routes.  

\section{Performance Evaluation}
\label{sec:sim}

\subsection{Performance Evaluation through Offline Simulations}
In this section, we evaluate CoRBA and CoRBA-fast through offline simulations. 

\subsubsection{Simulation Setup} 
\label{subsubsec:sim-setup}
In a single run of the offline simulation, the proposed algorithms are called to schedule a random coflow on a data center network.
\\ \noindent {\bf Data center network.}
For the network, we use a modified FatTree~\cite{fattree} architecture. 
A $k$-array FatTree network has $k$ pods, where each pod has $k$/2 Top-of-Rack (ToR) switches and $k$/2 aggregation switches. 
$(k/2)^2$ core switches are used connect the aggregation switches and each ToR switch connects a set of $k$/2 hosts.  
To better evaluate our algorithms, we increase the oversubscription ratio and therefore introduce more
competition on bandwidth at the aggregation and core levels. 
To achieve this goal, we multiply the number of hosts in a rack by a factor $\alpha_{over}$.
By doing so, a $k$-array modified FatTree contains $\alpha_{over} \cdot k^3/4$ hosts now.
In our simulations, we set $\alpha_{over}$ as 2 and set the each link's capacity as 10 Gbps.

\noindent {\bf Noise flows.} We introduce noise flows to simulate the complex traffic condition in real data center. 
Specifically, in each simulation, a set of $\alpha_{over} \cdot k^3$ noise flows are generated. 
The source and destination of each noise flow is randomly selected and the duration follows an uniform distribution in $[1,150]$.  
We randomly select a shortest path as its route and allocate a random amount of
bandwidth along that route. 
\smallskip \\ \noindent
{\bf Coflow.}
We randomly generate a coflow with $N$ flows. 
For each flow, we randomly select two hosts as its source and destination.
We further set the maximum possible volume of a flow (denoted by $V_{max}$) as 1000 Gb and determines the volume of a flow (i.e., $V_i$) by
using an uniform distribution in $[\beta\cdot V_{max}, V_{max}]$. 
In our simulations, we set $\beta$ as 0.7.

Each data point in our simulation results is an average of 20 simulations performed on an Intel 2.5 GHz processor.

\begin{figure*}[!t]
\begin{center}
\subfigure[Coflow Completion Time]{
\includegraphics[width=1.65in]{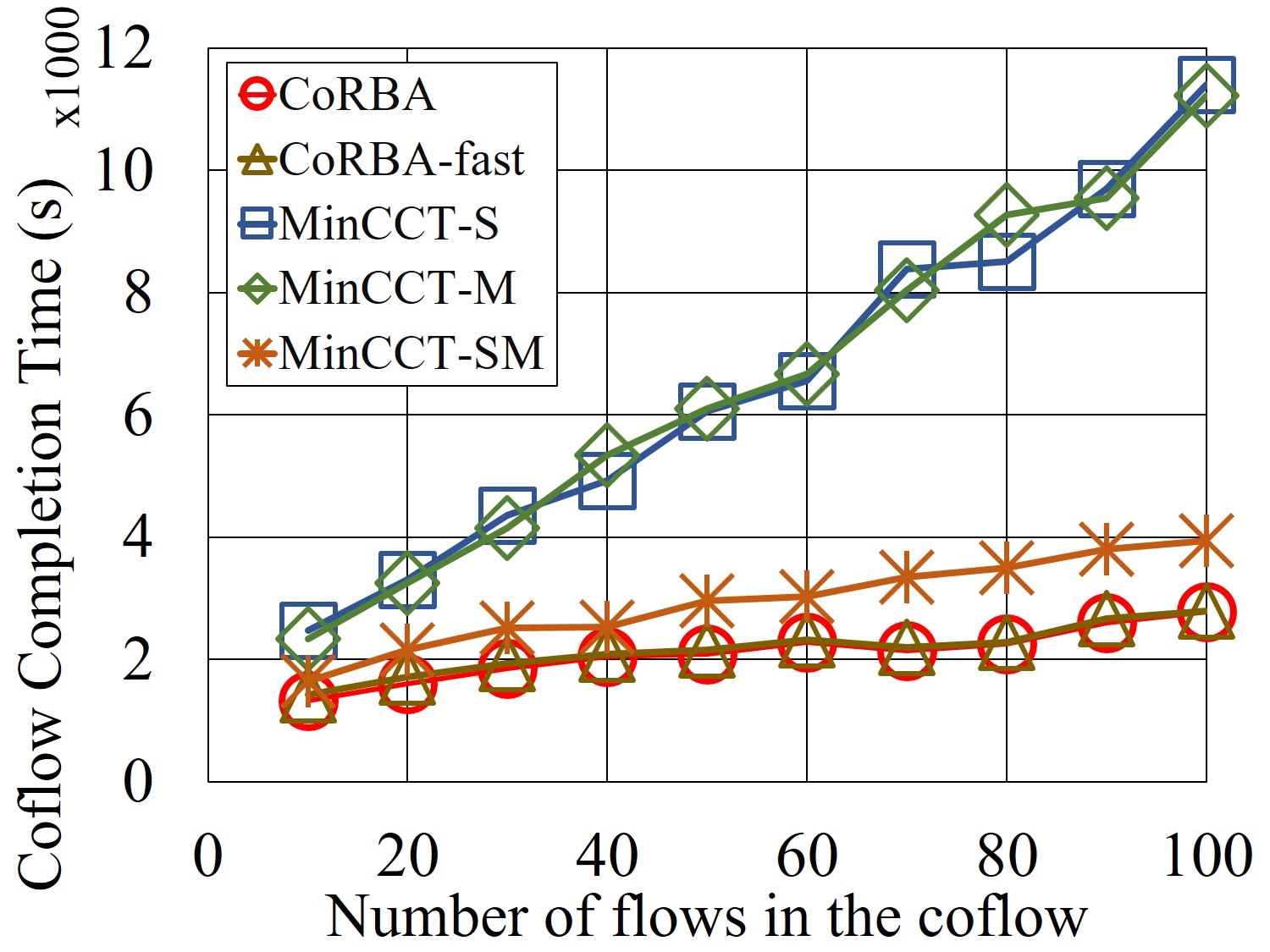}
\label{fig:cos-incflow-cct}
}
\subfigure[Allocated bandwidth]{
\includegraphics[width=1.65in]{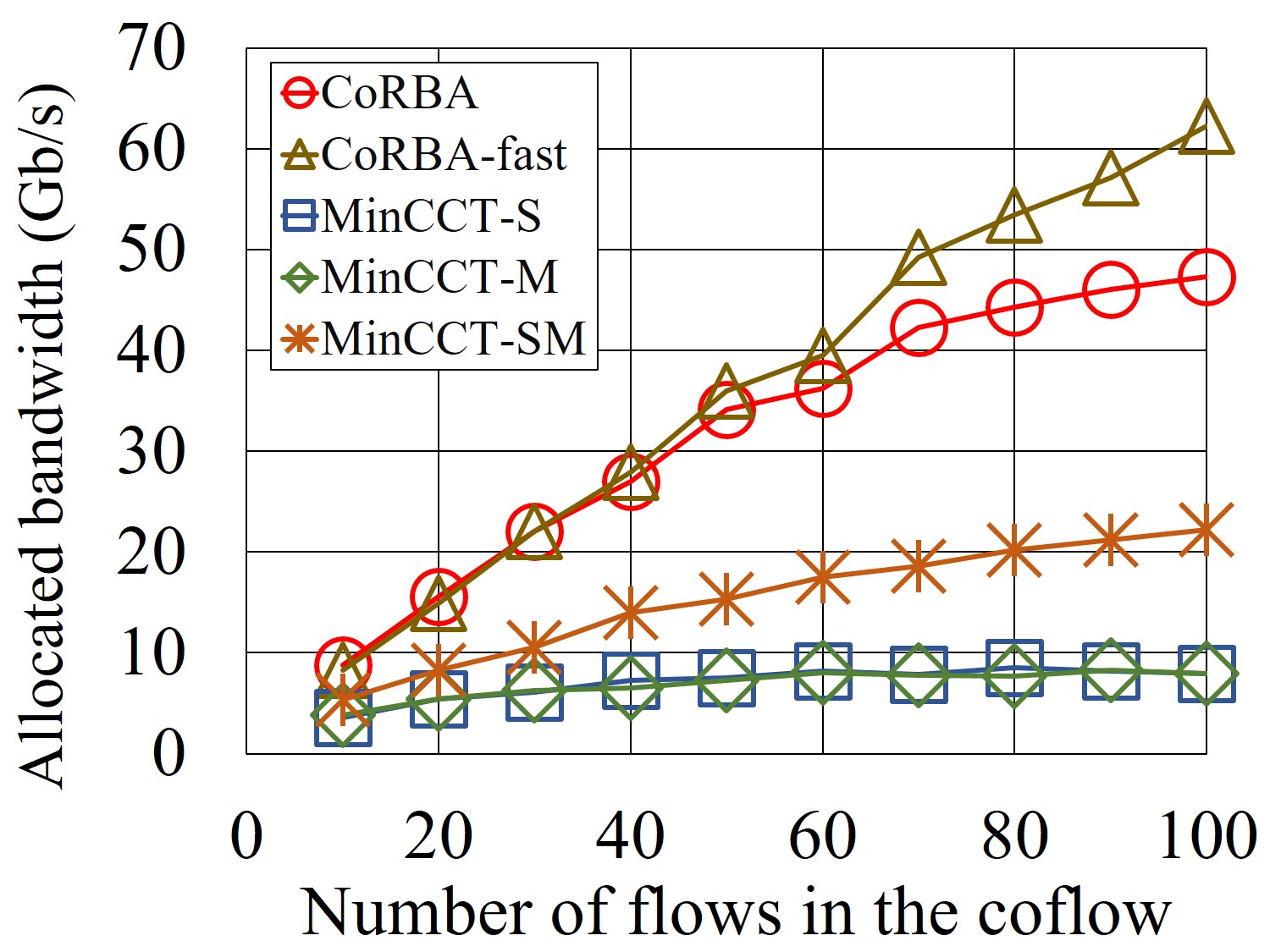}
\label{fig:cos-incflow-bw}
}
\subfigure[Average length of flow route]{
\includegraphics[width=1.65in]{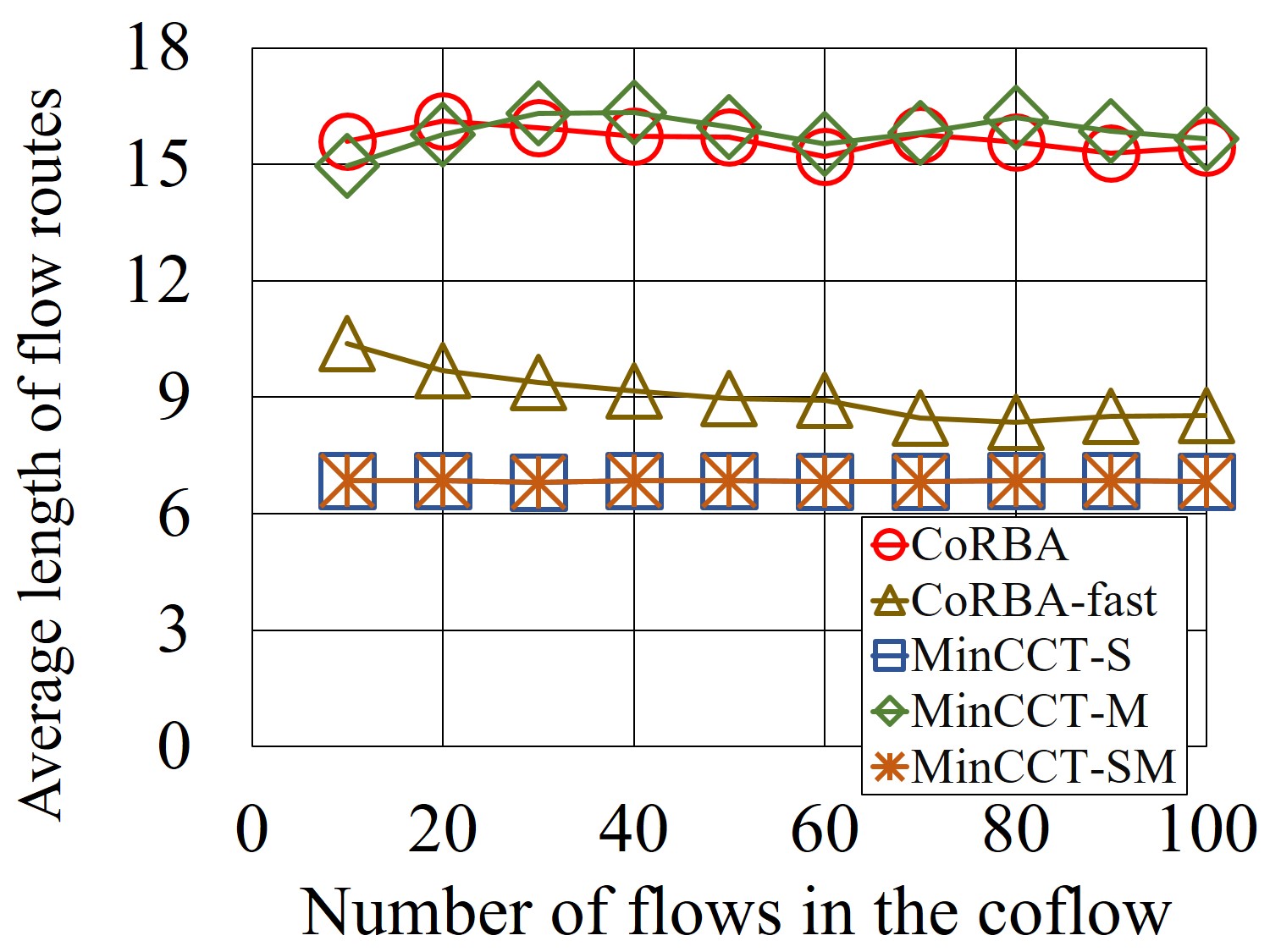}
\label{fig:cos-incflow-pathlen}
}
\subfigure[Running time]{
\includegraphics[width=1.65in]{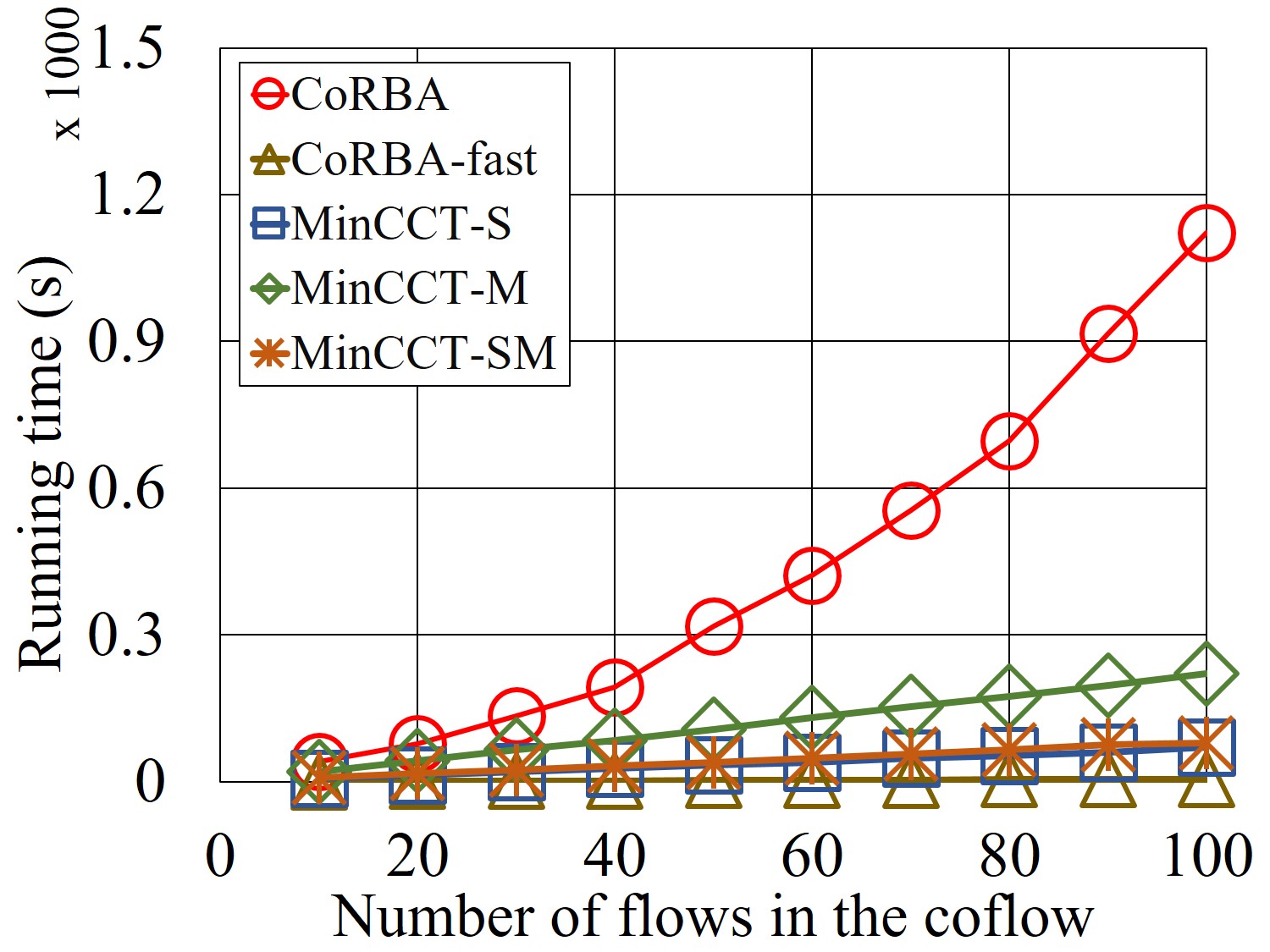}
\label{fig:cos-incflow-runtime}
}
\vspace{-2mm}
\caption{Performance of CoRBA and CoRBA-fast when scheduling different numbers of flows in a FatTree network with 1617 nodes.}
\label{fig:cos-incflow}
\vspace{-5mm}
\end{center}
\end{figure*}

\begin{figure*}[!t]
\begin{center}
\subfigure[Coflow Completion Time]{
\includegraphics[width=1.65in]{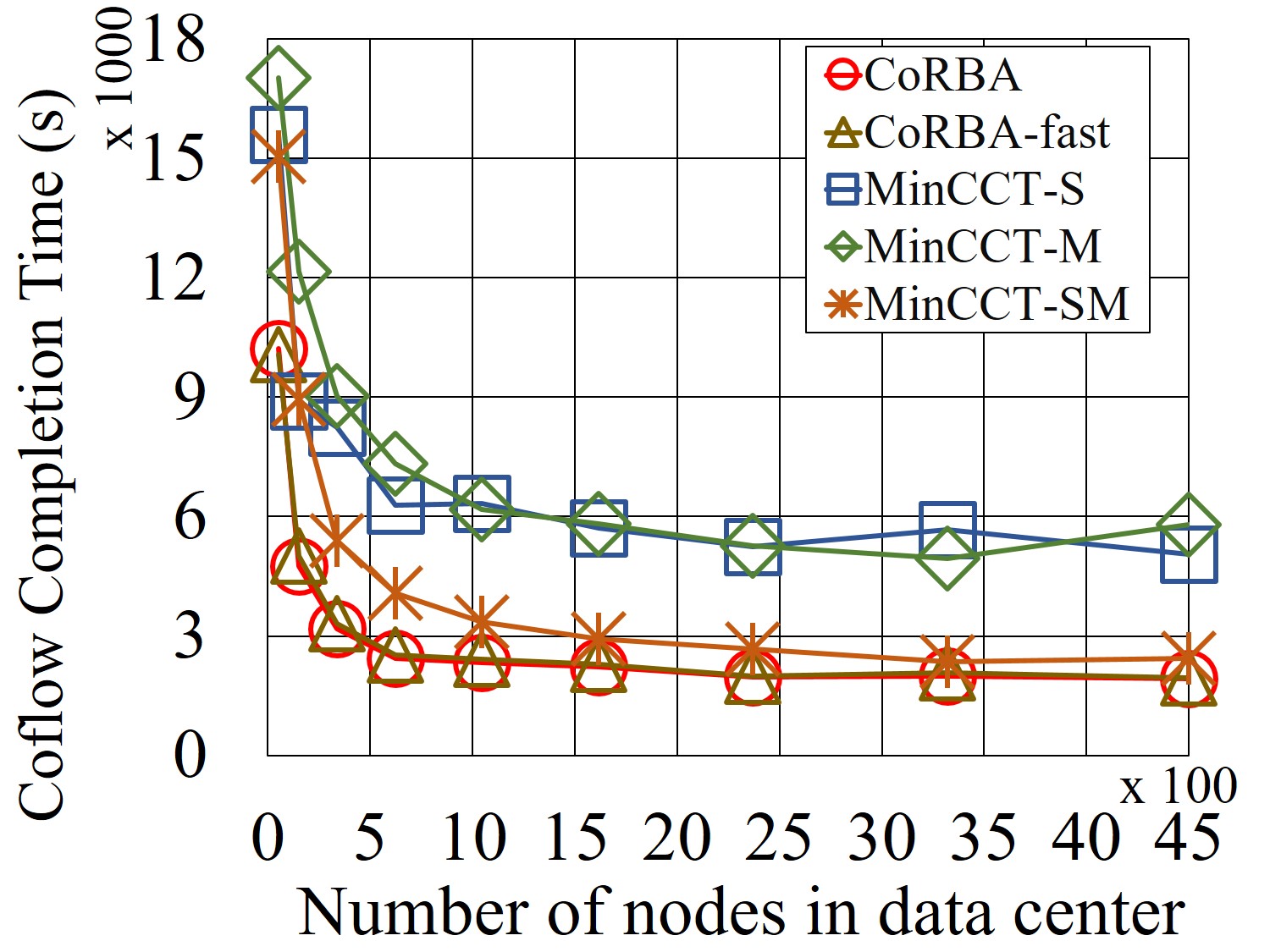}
\label{fig:cos-incnet-cct}
}
\subfigure[Allocated bandwidth]{
\includegraphics[width=1.65in]{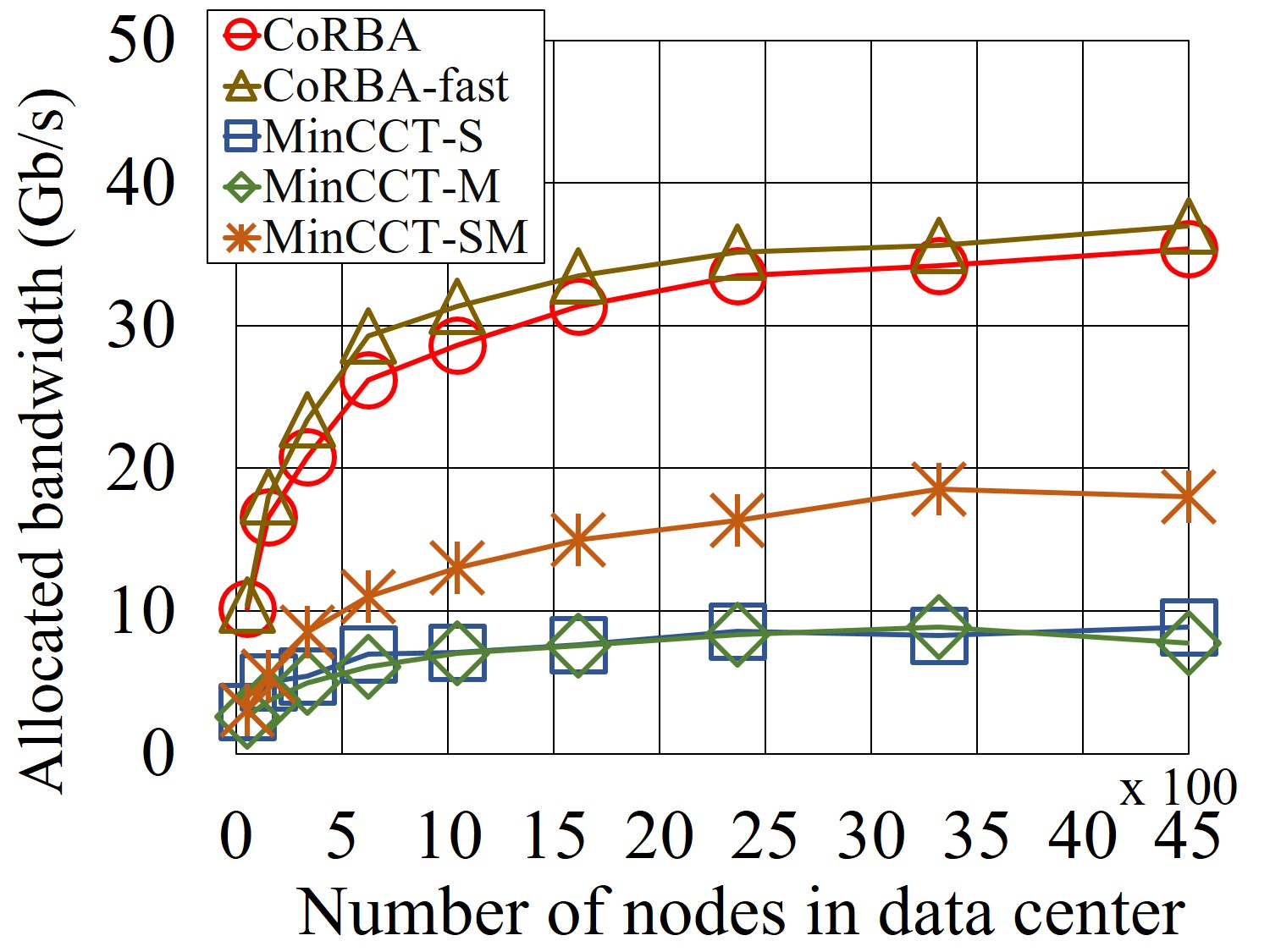}
\label{fig:cos-incnet-bw}
}
\subfigure[Average length of flow route]{
\includegraphics[width=1.65in]{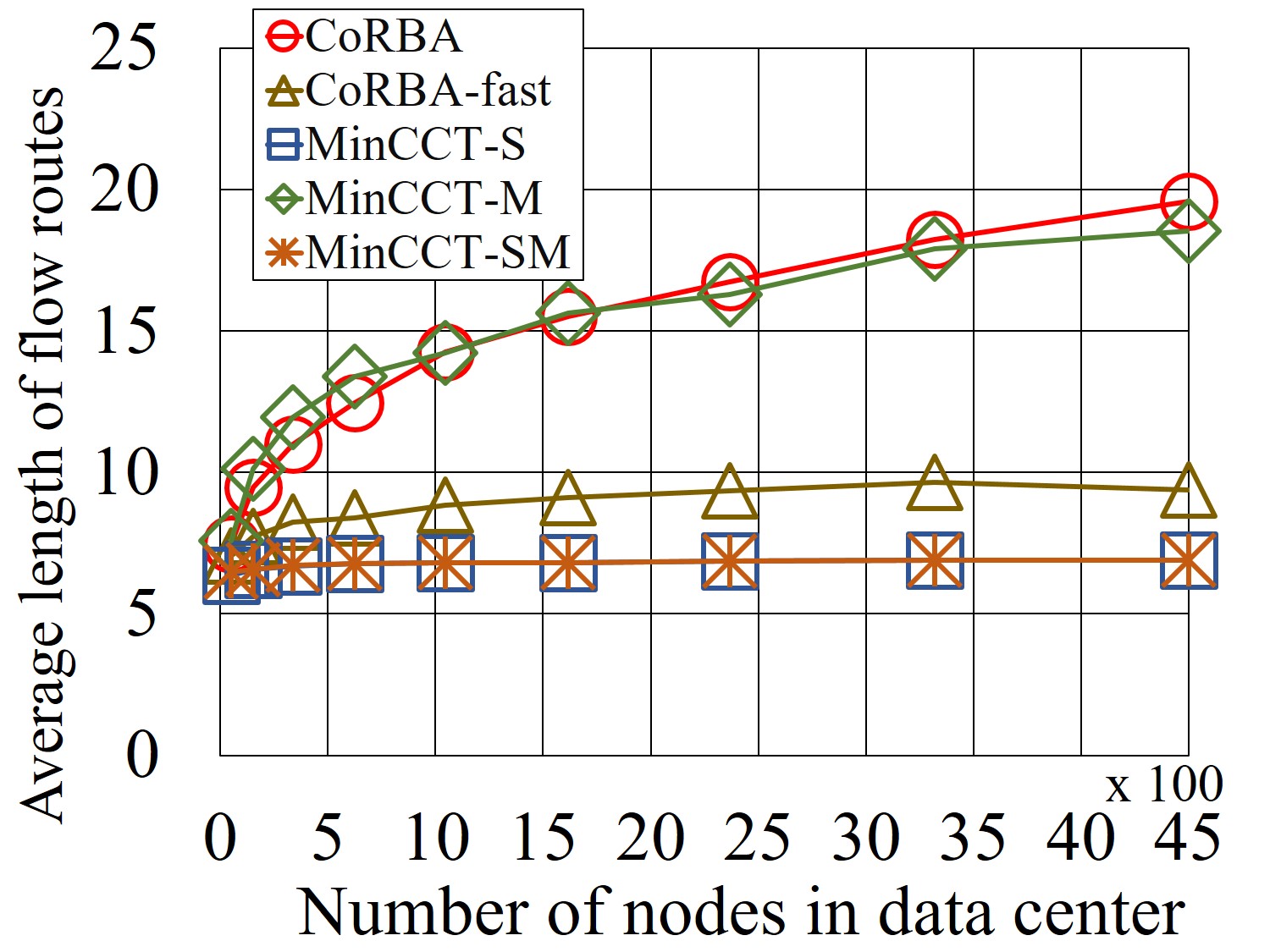}
\label{fig:cos-incnet-pathlen}
}
\subfigure[Running time]{
\includegraphics[width=1.65in]{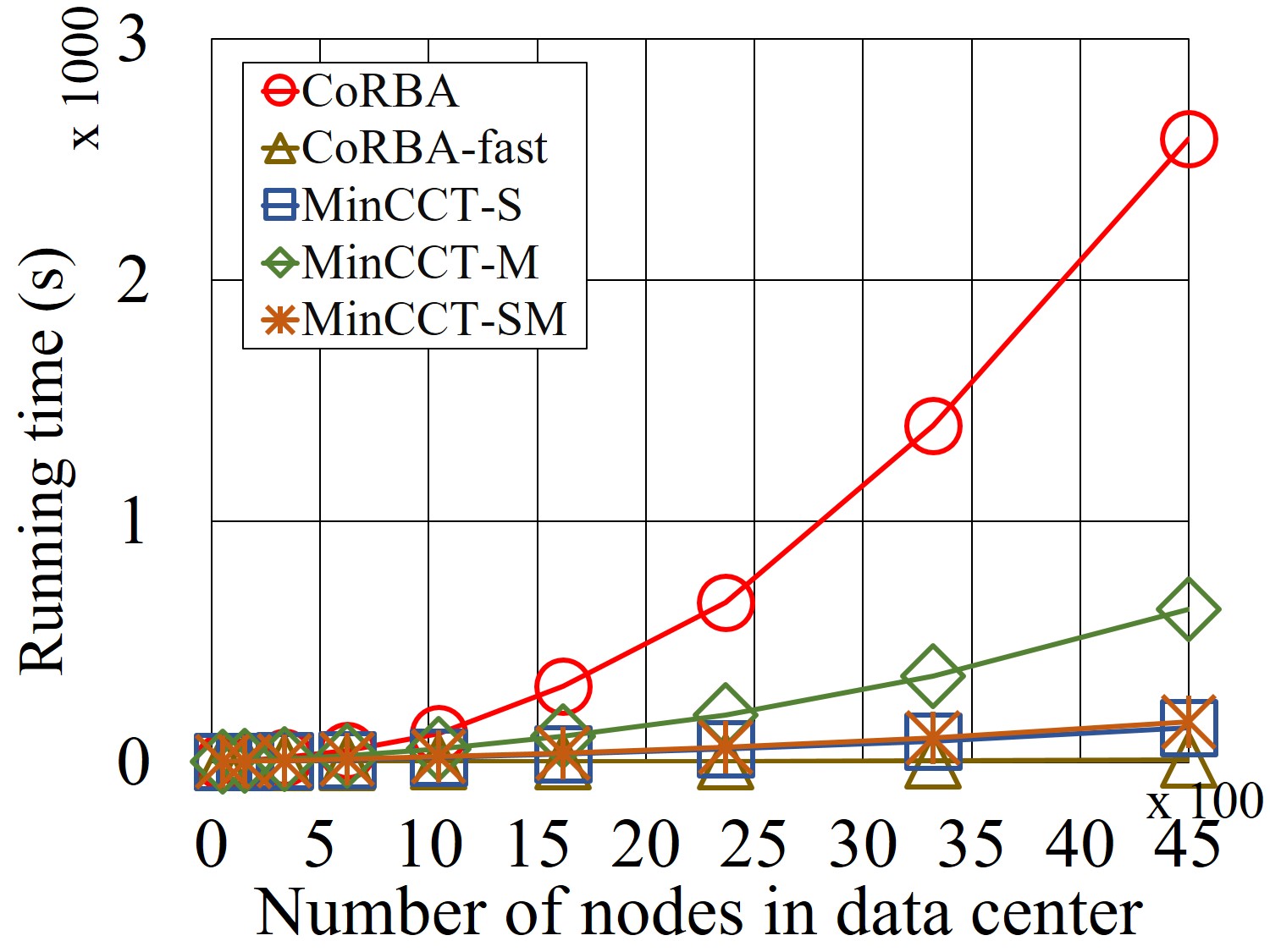}
\label{fig:cos-incnet-runtime}
}
\vspace{-2mm}
\caption{Performance of CoRBA and CoRBA-fast when scheduling a coflow with 50 flows in FatTree networks with different number of nodes.}
\label{fig:cos-incnet}
\vspace{-7mm}
\end{center}
\end{figure*} 

\subsubsection{Evaluation Metrics} 
We use four evaluation metrics. 
\noindent {\bf Coflow Completion Time.} This metric is the objective of the CoS problem. Hence, it is the most important metric. 
\noindent {\bf Allocated bandwidth.} This metric is the overall bandwidth allocated to $\mathcal{CF}$, which equals to $\sum^N_{i=1} b_i$.
It demonstrates how an algorithm performs on the aspect of bandwidth allocation. 
\noindent {\bf Average length of flow route.} This metric is the average number of hops on the route of flows in $\mathcal{CF}$.
It gives us a sense about how an algorithm performs on the aspect of routing. 

\noindent {\bf Running time.} Running time of an algorithm is also important. It gives a sense of the scalability of that algorithm.

\subsubsection{Comparison Algorithms}
As a comparison algorithm of CoRBA, we use a modified version of the MinimizeCCT
algorithm used in RAPIER~\cite{coflow:2015rapier}.
Given a coflow, MinimizeCCT selects a route from a set of candidates for each flow and allocates bandwidth to these flows while minimizing the CCT.
Although MinimizeCCT looks similar to the CoRBA algorithm, it does not perform true routing for the coflow. 

We modify MinimizeCCT to generate $K$ paths as candidates for each flow before scheduling the coflow.
We name this modified version as {\bf MinCCT}. 
We further propose three variants of MinCCT in which the $K$ potential paths are (i) $K$ shortest paths,
(ii) $K$ maximum capacity paths, and (iii) $K$ shortest maximum capacity paths.
We denote them by ``MinCCT-S'', ``MinCCT-M'', and ``MinCCT-SM'' respectively.
In our simulations, we set the value of $K$ as 5.


\subsubsection{Evaluation Results of CoRBA} 
\smallskip \noindent {\bf Performance with Increasing Size of the Coflow.}
In this simulation, we study how CoRBA and CoRBA-fast perform as the number of flows in the coflow increases from 10 to 100.
We uses a 16-array modified FatTree with $\alpha_{over} = 2$, which contains 1617 nodes. 

Fig.~\ref{fig:cos-incflow-cct} shows the CCT generated by each algorithm. 
While the CCT generally increases along with the expansion of the coflow,
CoRBA generates the smallest CCT. 
When the coflow contains 100 flows, the CCT of CoRBA is 300\% smaller than that of MinCCT-M and MinCCT-S, and 40\% smaller than that of MinCCT-SM. 
Meanwhile, CoRBA-fast generates almost the same CCT compared to CoRBA.

Fig.~\ref{fig:cos-incflow-bw} shows the total bandwidth allocated to the coflow. 
As can be seen, CoRBA allocates about 110\%-500\% more bandwidth than the MinCCT algorithms but about 25\% less bandwidth than CoRBA-fast.
Fig.~\ref{fig:cos-incflow-pathlen} shows the average length of flow route. 
We observe that CoRBA and MinCCT-M generates the longest flow route length, followed by CoRBA-fast, MinCCT-SM and MinCCT-S.

Putting Figs.~\ref{fig:cos-incflow-cct},~\ref{fig:cos-incflow-bw} and~\ref{fig:cos-incflow-pathlen} together, we can see that when the size of the
coflow increases, CoRBA and CoRBA-fast are able to allocate more bandwidth to the coflow and thereby
keep the growth of CCT relatively flat, as observed in Fig.~\ref{fig:cos-incflow-cct}.
We attribute this to their ability of routing based on the bandwidth availability of the whole network.
When the coflow expands, CoRBA and CoRBA-fast are able to route flows via different paths and utilize more bandwidth.
In contrast, the MinCCT algorithms are restricted by the limited number of candidate paths and cannot sufficiently utilize the available resources in
the network, which leads the rapid increase of the CCT.

At last, Fig.~\ref{fig:cos-incflow-runtime} shows the algorithm running time. When the coflow contains 100 flows, the running time of CoRBA is about
1100 seconds, while that of CoRBA-fast, MinCCT-S, MinCCT-M, and MinCCT-SM is 5 seconds, 70 seconds, 220 seconds, and 80 seconds respectively.


\begin{figure*}[!t]
\begin{center}
\subfigure[Coflow Completion Time]{
\includegraphics[width=1.65in]{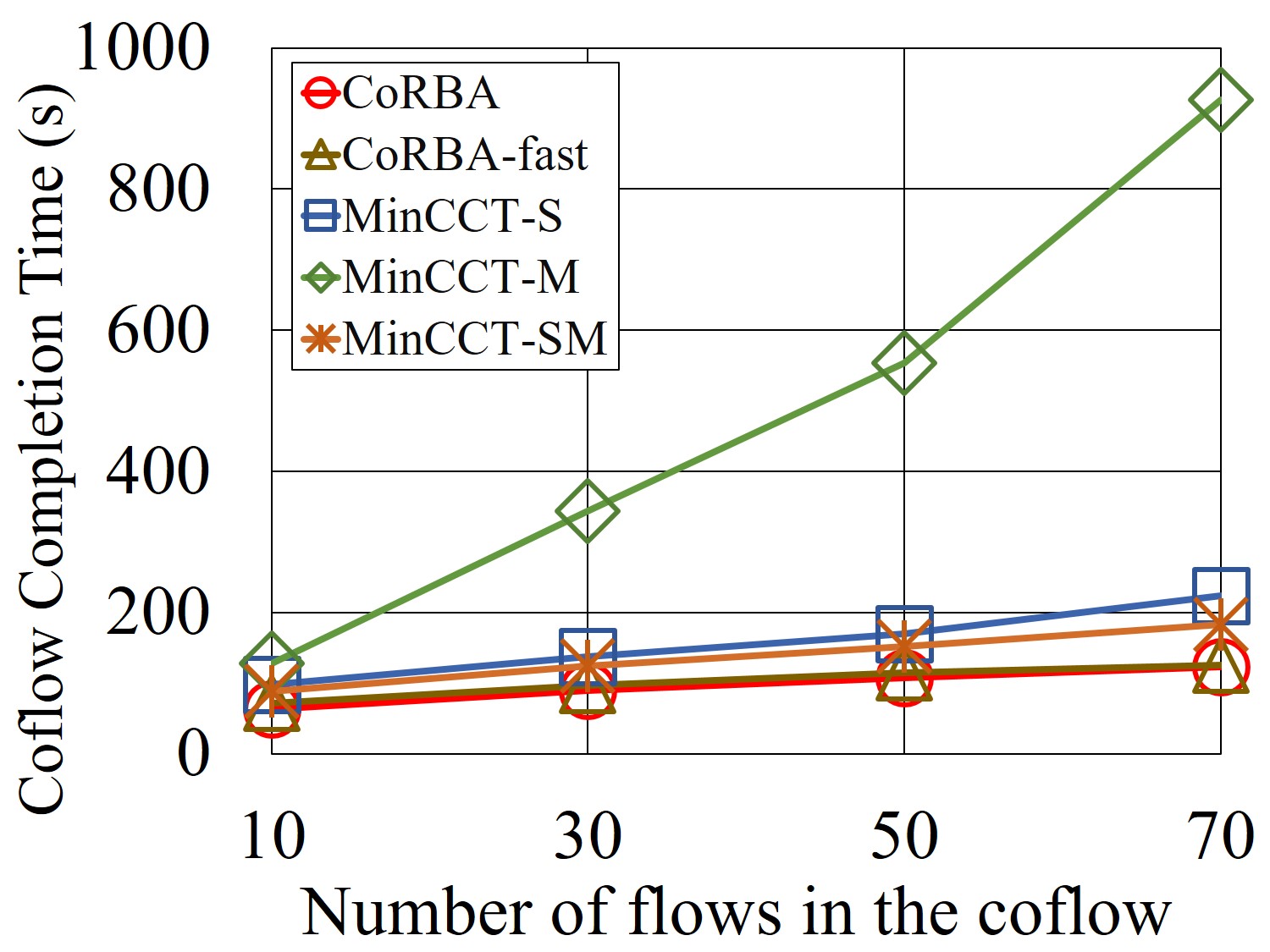}
\label{fig:cos-online-incflow-cct}
}
\subfigure[Ratio CCT to CoRBA]{
\includegraphics[width=1.65in]{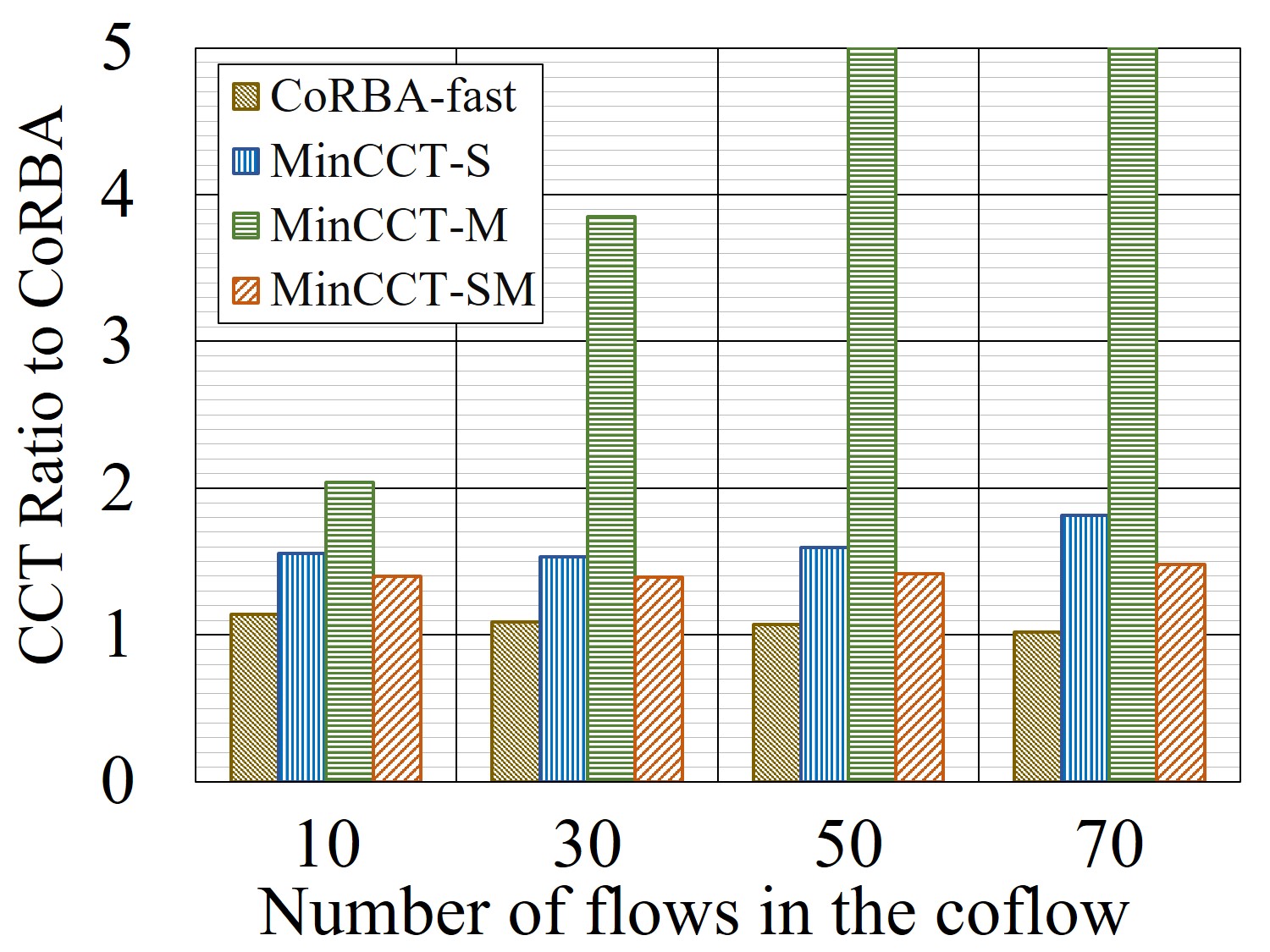}
\label{fig:cos-online-incflow-ratio}
}
\subfigure[Running time]{
\includegraphics[width=1.65in]{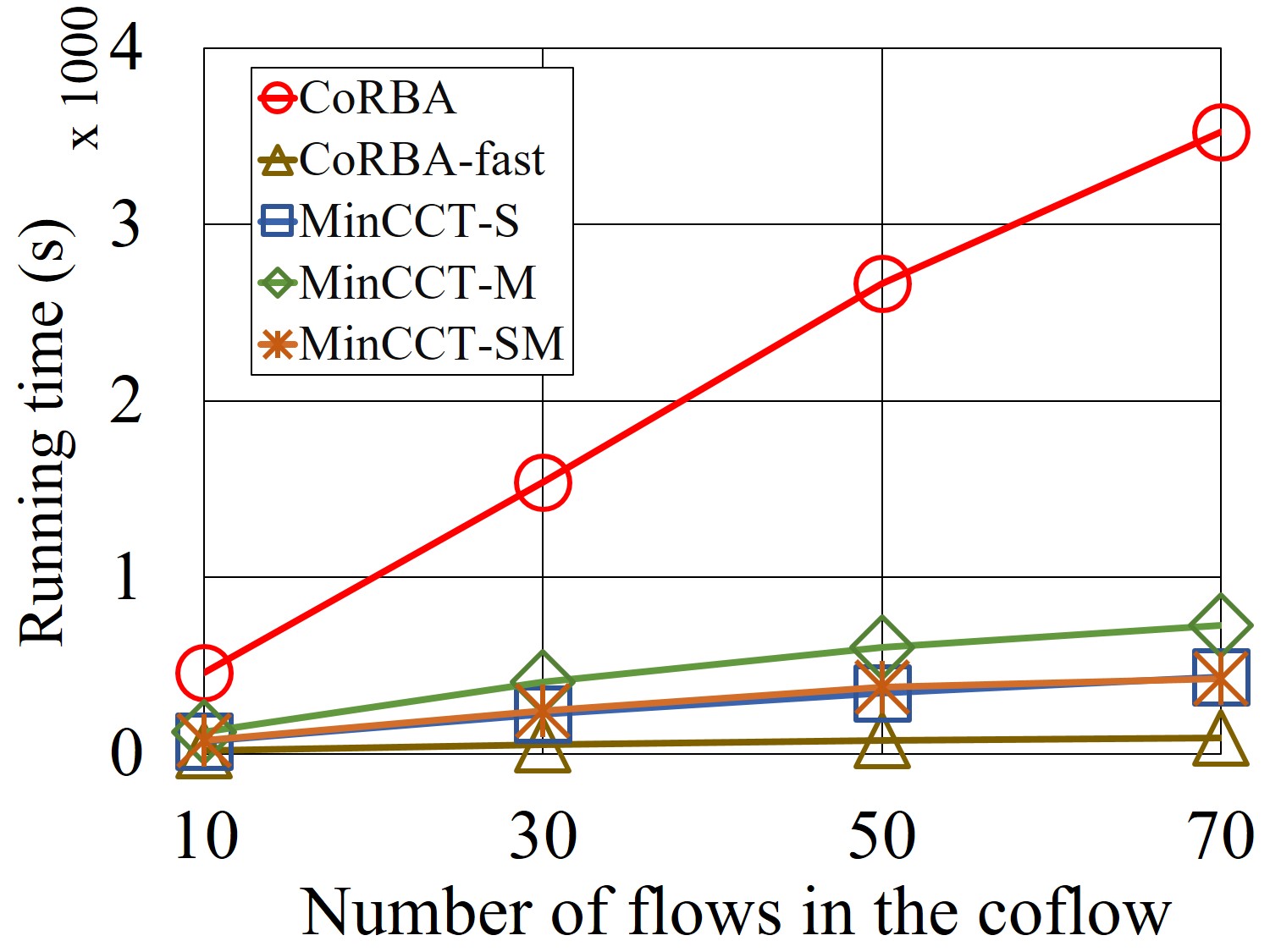}
\label{fig:cos-online-incflow-runtime}
}
\vspace{-2mm}
\caption{Performance of CoRBA and CoRBA-fast in online simulations with increasing number of flows in the coflow.}
\label{fig:cos-online-incflow}
\vspace{-7mm}
\end{center}
\end{figure*}

\smallskip \noindent {\bf Performance with Increasing Size of the Network.}
In this simulation, we demonstrate how our algorithms perform as the number of nodes in the FatTree network increases from 52 to 4500, i.e.,
the number of pods increases from 4 to 20. 
We fix the number of flows in the coflow at 50. 

Fig.~\ref{fig:cos-incnet-cct} shows that the CCT of CoRBA is similar to that of CoRBA-fast, 30\% smaller than that of MinCCT-SM, and 200\% smaller
than that of MinCCT-S and MinCCT-M.
Fig.~\ref{fig:cos-incnet-bw} shows that CoRBA allocate 100\% more bandwidth than MinCCT-SM, 300\% more bandwidth than MinCCT-M and MinCCT-SM, but
similar amount to CoBRA-fast.

We observe that the bandwidth allocated by CoRBA is dramatically increased when the network just starts expanding and becomes more smooth
thereafter.
We attribute this to the limitation on how much bandwidth the algorithm can find to utilize.
When the network is small, the number of good paths (with large available bandwidth) is also small. 
Consequently, CoRBA may schedule some flows to use paths with less bandwidth, which limits the CCT. 
At this stage, an expansion of network generates more good paths and therefore CoRBA can allocate more bandwidth. 
However, along with the expansion, the number of good paths becomes more than enough and CoRBA is able to schedule most of the flows on these
paths.
At this stage, the quality of paths is barely improved.
As a result, the increase of total allocated bandwidth becomes small. 
Such a trend then leads the trend of CCT, which is a steep decrease at the beginning but becomes relatively flat
thereafter, as shown in Fig.~\ref{fig:cos-incnet-cct}.


Fig.~\ref{fig:cos-incnet-pathlen} shows the average length of flow routes.
MinCCT-M and CoRBA have the longest route length, followed by CoRBA-fast, MinCCT-S and MinCCT-SM.
At last, Fig.~\ref{fig:cos-incnet-runtime} shows the algorithm running time. 
When the network contains 4500 nodes, the running
time of CoRBA is about 2500 seconds, while that of CoRBA-fast, MinCCT-S, MinCCT-M, and MinCCT-SM is 10 seconds, 140 seconds, 600 seconds, and 160 seconds respectively.

{\bf Summary.} In offline simulations, we examined the performance of CoRBA and CoRBA-fast when scheduling a single coflow. 
Benefiting from the ability of finding paths with less overlaps and allocating more bandwidth, CoRBA and CoRBA-fast outperforms MinCCT-SM by a factor
of 1.3-1.5 and outperforms MinCCT-S and MinCCT-M by a factor of 3-5. 
Moreover, CoRBA-fast has similar performance with CoRBA but less efficiency on utilizing available bandwidth.

\begin{figure*}[!t]
\begin{center}
\subfigure[Coflow Completion Time]{
\includegraphics[width=1.65in]{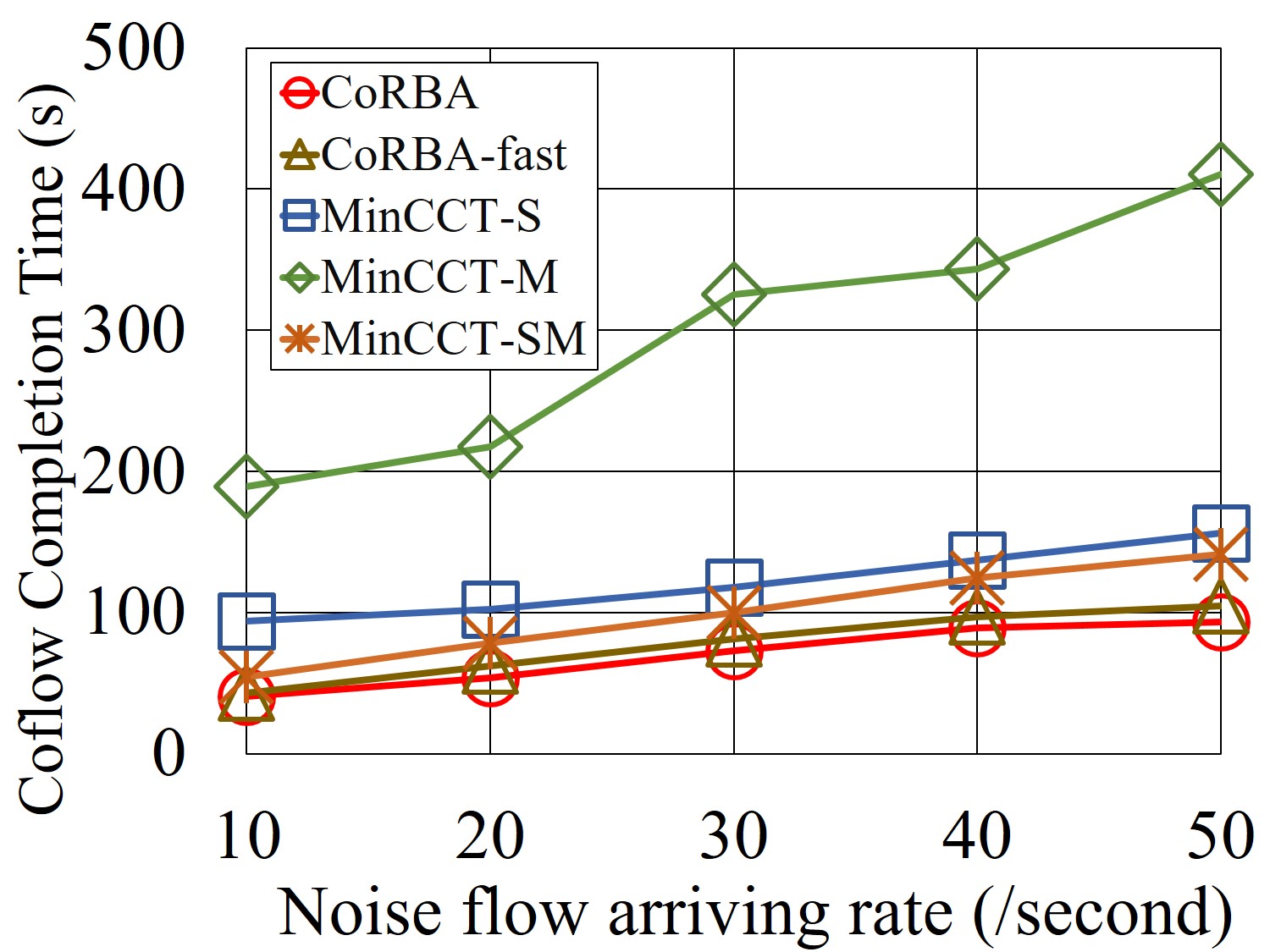}
\label{fig:cos-online-incnoise-cct}
}
\subfigure[Ratio CCT to CoRBA]{
\includegraphics[width=1.65in]{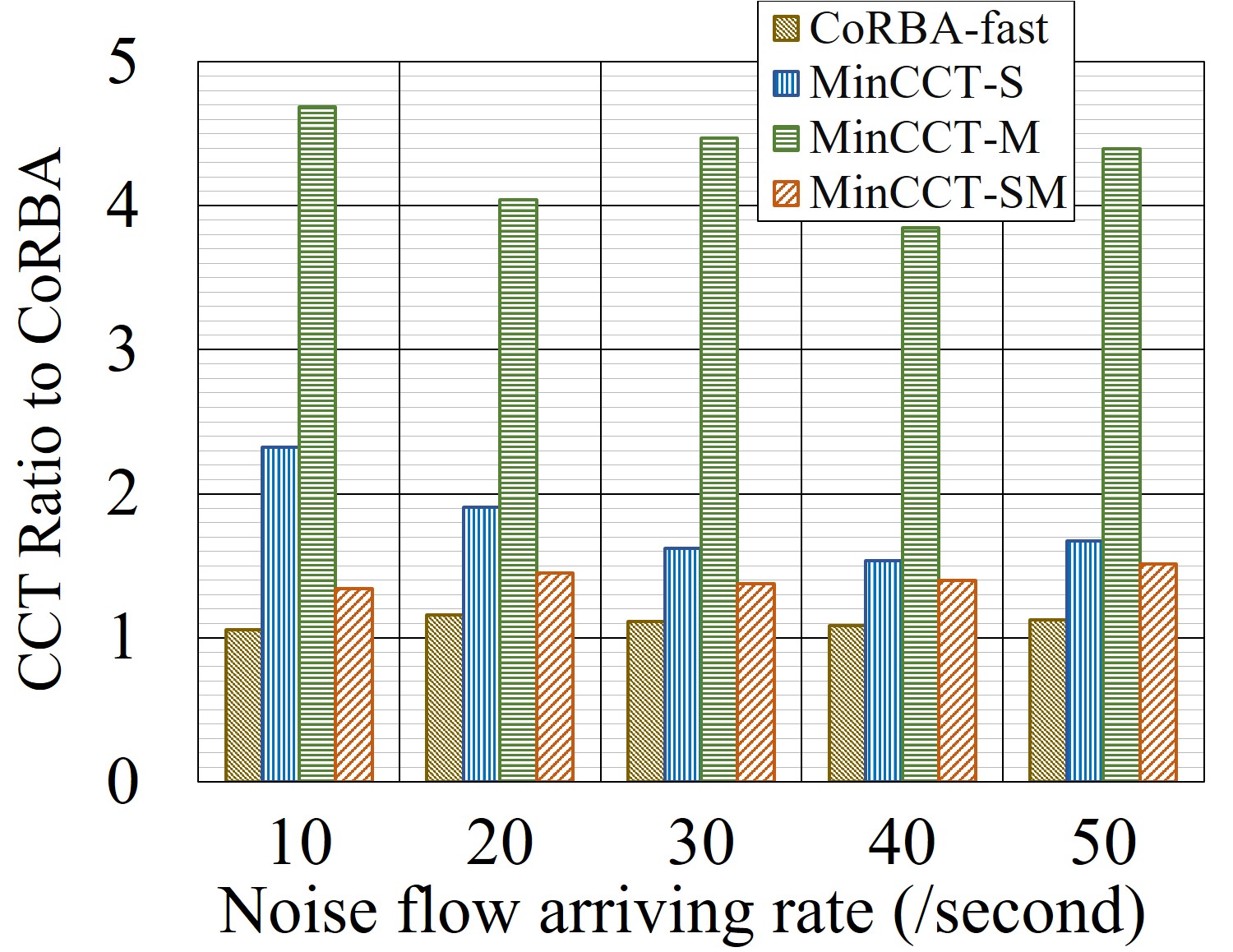}
\label{fig:cos-online-incnoise-ratio}
}
\subfigure[Running time]{
\includegraphics[width=1.65in]{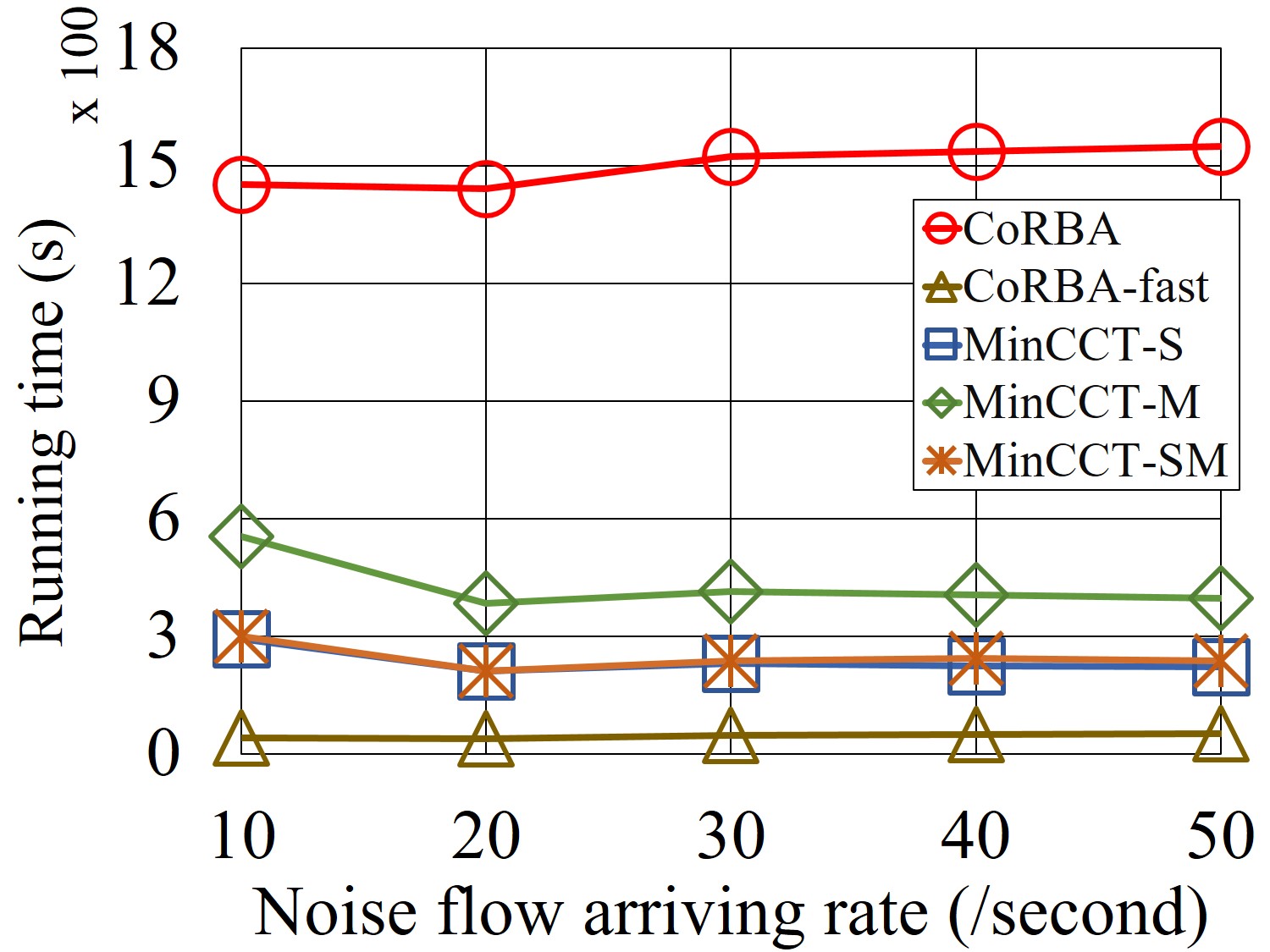}
\label{fig:cos-online-incnoise-runtime}
}
\vspace{-2mm}
\caption{Performance of CoRBA and CoRBA-fast in online simulations with increasing arriving rate of noise flows.}
\label{fig:cos-online-incnoise}
\vspace{-7mm}
\end{center}
\end{figure*}

\subsection{Performance Evaluation through Online Simulations}
In this section, we examine long-term performance of the proposed algorithms through online simulations in which we simulate the execution of a flow
scheduler using the state-of-the-art scheduling policy.
\subsubsection{Simulation Setup}
A simulation begins with a 10-array (625 nodes) modified FatTree network without any flows. 
In the following, coflows start to arrive at a rate following a Poisson distribute with $\mu = 0.01/\text{second}$, and stop arriving at
1800 seconds.
The whole simulation finishes, once all coflows are scheduled. 
Each data point in the results is an average of 10 simulations performed on an Intel 2.5GHz processor.
\smallskip \\
\noindent {\bf Scheduling Policy.}
We use the scheduling policy proposed in RAPIER~\cite{coflow:2015rapier} to schedule arrived coflows.  
Once a coflow arrives, the scheduling algorithm is called to calculate the CCT of all existing coflows based on their current
residual volume.
In the following, all coflows are scheduled in the ascending order of their CCT. 
Once a coflow is scheduled, the CCT of all other coflows are re-calculated based on updated bandwidth availability. 
The selection procedure is repeated until all existing coflows are addressed.
An existing coflow may be preempted by a newly arrived coflow.
If scheduling a coflow is failed, this coflow is added into a waiting queue. 
Meanwhile, if a coflow has been waiting for more than 100 seconds, it gets the privilege to be scheduled regardless of its calculated CCT.  
At last, once a coflow finishes,the released bandwidth is distributed to existing coflows, aiming to facilitate the flow completion. 
\smallskip \\ \noindent {\bf Noise flows.}
The noise flows arrives following a Poisson distribution 
and the duration of these noise flows
follows an uniform distribution in $[1,150]$.
When a noise flow starts, we randomly selects a shortest path as its route and allocates a random amount of bandwidth. 

\subsubsection{Evaluation Metrics}
We focus on two metrics: (i) the average CCT of coflows and (ii) the algorithm running time.


\subsubsection{Evaluation Results of CoRBA}
\smallskip \noindent {\bf Performance with Increasing Size of the Coflow.}
We first study how our algorithms perform when the number of flows in the coflow increases. 
We set the arriving rate of noise flows to follow a Poisson distribution with $\mu = 40/\text{second}$.

Fig.~\ref{fig:cos-online-incflow-cct} shows the average CCT generated by each algorithm and Fig.~\ref{fig:cos-online-incflow-ratio} shows the ratio of
the average CCT of other algorithms to that of CoRBA. 
We observe that the average CCT generated by CoRBA is about 10\%, 40\%, 60\%-80\%, and 100\%-650\% smaller than that generated by CoRBA-fast,
MinCCT-SM, MinCCT-S, and MinCCT-M respectively.
Fig.~\ref{fig:cos-online-incflow-runtime} shows the running time of each algorithm. 
While CoRBA has the largest running time, CoRBA-fast has the smallest running time which is 4-8 times faster than the three MinCCT algorithms and
about 40 times faster than CoRBA .
\smallskip \\ \noindent {\bf Performance with Increasing Size of the Coflow.}
We further study how our algorithms perform when the arriving rate of noise flows increases. 
We set each coflow containing 30 flows. 

Fig.~\ref{fig:cos-online-incnoise-cct} shows the average CCT generated by each algorithm and Fig.~\ref{fig:cos-online-incnoise-ratio} shows the ratio
of the average CCT of other algorithms to that of CoRBA.
We can see that the average CCT generated by CoRBA is about 10\%, 40\%, 60\%--120\%, and 400\%-500\% smaller than that generated by CoRBA-fast,
MinCCT-SM, MinCCT-S, and MinCCT-M respectively.
Fig.~\ref{fig:cos-online-incnoise-runtime} shows the running time of each algorithm. 
Again, CoRBA has the largest running time and CoRBA-fast has the smallest running time which is 4-8 times faster than the three MinCCT algorithms
and around 30 times faster than CoRBA. 

{\bf Summary.} In online simulations, we demonstrate that CoRBA and CoRBA-fast maintain their advantage over the MinCCT algorithms:
They are at least 40\% (and up to 500\%) better than the comparison algorithms.
We attribute this to their better performance when scheduling each individual coflow.
Meanwhile, CoRBA-fast is much faster than CoRBA but CoRBA performs about 10\% better than CoRBA-fast. 

\section{Discussion}
\label{sec:discuss}
In simulations, it has been shown that CoRBA-fast can be tens of times faster than other algorithms, which makes it a good choice in practice.  
In this section, we further discuss about its applicability by comparing its running time with the transferring time of flows in some typical
MapReduce jobs. 

The shuffle stage in MapReduce jobs generates coflows. In shuffle-heavy jobs, like tera-sort and ranked-inverted-index, the shuffle
volume can be more than 200 GB~\cite{mapreduce:2013ishuffle,mapreduce:2014mronline}. 
Considering the simulation case that scheduling a coflow with 100 flows in a network with 1617 nodes, the running time of CoRBA-fast is
about 6 seconds and the allocated bandwidth is around 60 Gb/s.
In this case, if the coflow contains 200 GB data, its transferring time is about 27 seconds, which is nearly 5 times of the running time of
CoRBA-fast. 
Whereas, the MinCCT algorithms allocate 10-20 Gb/s bandwidth and the coflow transferring time is 80-160 seconds.
We can see that CoRBA-fast significantly reduces the coflow transferring time and its running time is generally small compared to the reduced
transferring time. 
When examining other simulation cases, we get similar results.
Therefore, we believe that the use of CoRBA-fast is very applicable in practice. 
\section{Conclusion}
\label{sec:conlusion}
In this paper, we focused on how to schedule---route and allocate bandwidth to---a coflow with the goal of minimizing the its completion time. 
We first studied the problem of optimal bandwidth allocation with pre-determined routes and provided an analytical solution. 
Subsequently, we formulated the coflow scheduling problem as a NLMIP problem, presented a relaxation together with an equivalent
convex optimization problem, and proposed an algorithm called CoRBA and a simplified version called CoRBA-fast to solve the problem.

We evaluated the CoRBA and CoRBA-fast algorithms by comparing them with the state-of-the-art algorithms in both offline and online simulations. 
In offline simulations, we demonstrated that when scheduling a single coflow, CoRBA and CoRBA-fast 30\%-400\% smaller CCT than their comparison
algorithms.
In online simulations, we simulated the execution of a flow scheduler and used the state-of-the-art scheduling policy. 
The results show that CoRBA and CoRBA-fast are at least 40\% and up to 500\% better than the comparison algorithms. 
Meanwhile, the results also show that CoRBA-fast can be tens of times faster than all other algorithms with the cost of about 10\% performance
degradation compared to CoRBA, which makes CoRBA-fast very applicable in practice.

\bibliographystyle{IEEEtran}
\bibliography{IEEEabrv,reference}

\begin{IEEEbiographynophoto}{Li Shi}
received his Ph.D. degree from Department of Electrical and Computer Engineering at Stony
Brook University, Stony Brook, NY, in 2016. Previous,  he received his B.E. degree in electrical and
computer engineering from Shanghai Jiao Tong University, Shanghai, China, in 2010.
He is working at Snap Inc, Venice, CA.
His research interests include task scheduling and resource allocation in data centers, cloud
computing, data center network, software defined network, etc.
\end{IEEEbiographynophoto}
\begin{IEEEbiographynophoto}{Yang Liu}
received his Ph.D. degree from Department of Electrical and Computer Engineering at Stony
Brook University, Stony Brook, NY, in 2017. Previously, he received his B.E. degree from Department
of Electrical and Computer Engineering at University of Electronic Science and Technology of China,
Chengdu, China, in 2011. His research interests are in the area of distributed/parallel computing,
networking, and load balancing algorithms. He is currently working on divisible load theory and
heterogeneous system applications.
\end{IEEEbiographynophoto}
\begin{IEEEbiographynophoto}{Junwei Zhang}
Dr. Junwei Zhang received the PhD degree from the Applied Mathematics and Statistics Department of Stony Brook University in 2018.  His research interests include parallel computing optimization, computational geometry and applied machine learning.
\end{IEEEbiographynophoto}
\begin{IEEEbiographynophoto}{Thomas G. Robertazzi}
received the Ph.D from Princeton University, Princeton, NJ, in 1981 and the B.E.E.
from the Cooper Union, New York, NY in 1977. 
He is presently a Professor in the Dept. of Electrical and Computer Engineering at Stony Brook University, Stony Brook N.Y. 
He has published extensively in the areas of parallel processing scheduling, telecommunications and performance evaluation. 
Prof. Robertazzi has also authored, co-authored or edited six books in the areas of networking,
performance evaluation, scheduling and network planning. 
He is a Fellow of the IEEE and was co-chair of the Stony Brook University Senate Research
Committee from 2008 to 2018.
\end{IEEEbiographynophoto}

\end{document}